\pdfoutput=1
\RequirePackage{ifpdf}
\ifpdf 
\documentclass[pdftex]{sigma}
\else
\documentclass{sigma}
\fi

\numberwithin{equation}{section}
\newtheorem{Theorem}{Theorem}[section]
\newtheorem{Corollary}[Theorem]{Corollary}
\newtheorem{Proposition}[Theorem]{Proposition}
 { \theoremstyle{definition}
\newtheorem{Remark}[Theorem]{Remark} }

\usepackage{tikz-cd}

\def\od#1#2{\frac{{\rm d}#1}{{\rm d}#2}}
\def\pd#1#2{\frac{\partial #1}{\partial #2}}
\def\fd#1#2{\frac{\delta #1}{\delta #2}}

\def\tfd#1#2{\delta #1/\delta #2}

\def\dzero#1#2{\left.\od{}{#1} #2 \right|_{#1=0}}

\def\parentheses#1{\left(#1\right)}

\def\tr{\mathop{\mathrm{tr}}\nolimits}

\def\diag{\operatorname{diag}}

\def\norm#1{{\left\|#1\right\|}}
\def\abs#1{\left|#1\right|}
\def\DS{\displaystyle}
\def\R{\mathbb{R}}
\def\C{\mathbb{C}}

\def\defeq{\mathrel{\mathop:}=}

\def\setdef#1#2{\left\{ #1 \, |\, #2 \right\}}
\def\ip#1#2{{\left\langle#1,#2\right\rangle}}
\def\tip#1#2{{\langle#1,#2\rangle}}

\renewcommand{\Re}{\operatorname{Re}}
\renewcommand{\Im}{\operatorname{Im}}

\def\rmi{{\rm i}}

\def\d{\mathbf{d}}
\def\ins#1{{\bf i}_{#1}}
\def\PB#1#2{\left\{#1,#2\right\}}
\newcommand\Ad{\operatorname{Ad}}

\def\bPi{\boldsymbol{\Pi}}
\def\bGamma{\boldsymbol{\Gamma}}

\def\SO{\mathsf{SO}}
\def\SE{\mathsf{SE}}

\def\SU{\mathsf{SU}}

\def\so{\mathfrak{so}}
\def\se{\mathfrak{se}}

\def\su{\mathfrak{su}}

\begin{document}
\allowdisplaybreaks

\newcommand{\arXivNumber}{1907.07819}

\renewcommand{\PaperNumber}{083}

\FirstPageHeading

\ShortArticleName{Collective Heavy Top Dynamics}

\ArticleName{Collective Heavy Top Dynamics}

\Author{Tomoki OHSAWA}

\AuthorNameForHeading{T.~Ohsawa}

\Address{Department of Mathematical Sciences, The University of Texas at Dallas, \\
800 W Campbell Rd, Richardson, TX 75080-3021, USA}
\Email{\href{mailto:tomoki@utdallas.edu}{tomoki@utdallas.edu}}
\URLaddress{\url{http://www.utdallas.edu/~tomoki/}}

\ArticleDates{Received July 20, 2019, in final form October 22, 2019; Published online October 30, 2019}

\Abstract{We construct a Poisson map $\mathbf{M}\colon T^{*}\mathbb{C}^{2} \to \mathfrak{se}(3)^{*}$ with respect to the canonical Poisson bracket on $T^{*}\mathbb{C}^{2} \cong T^{*}\mathbb{R}^{4}$ and the $(-)$-Lie--Poisson bracket on the dual $\mathfrak{se}(3)^{*}$ of the Lie algebra of the special Euclidean group $\mathsf{SE}(3)$. The essential part of this map is the momentum map associated with the cotangent lift of the natural right action of the semidirect product Lie group $\mathsf{SU}(2) \ltimes \mathbb{C}^{2}$ on $\mathbb{C}^{2}$. This Poisson map gives rise to a canonical Hamiltonian system on $T^{*}\mathbb{C}^{2}$ whose solutions are mapped by $\mathbf{M}$ to solutions of the heavy top equations. We show that the Casimirs of the heavy top dynamics and the additional conserved quantity of the Lagrange top correspond to the Noether conserved quantities associated with certain symmetries of the canonical Hamiltonian system. We also construct a Lie--Poisson integrator for the heavy top dynamics by combining the Poisson map $\mathbf{M}$ with a simple symplectic integrator, and demonstrate that the integrator exhibits either exact or near conservation of the conserved quantities of the Kovalevskaya top.}


\Classification{37J15; 37M15; 53D20; 70E17; 70E40; 70H05; 70H06}

\section{Introduction}
\subsection{Heavy top equations}
Consider a rigid body spinning around a fixed point (\textit{not} the center of mass) in a gravitational field. We specify the configuration of the body by a rotation matrix around the fixed point: the rotation matrix defines a transformation from the body frame to the spatial frame; hence the original configuration space is $\SO(3)$.
Let $t \mapsto R(t) \in \SO(3)$ be the rotational dynamics of the body and define the body angular velocity by $\hat{\Omega}(t) \defeq R^{-1}(t)\dot{R}(t) \in \so(3)$.
One may identify it with a vector $\boldsymbol{\Omega}(t) \in \R^{3}$ via the standard identification between $\so(3)$ and $\R^{3}$ summarized in Appendix~\ref{sec:su2}.
Let $\mathbb{I} \defeq \diag(I_{1}, I_{2}, I_{3})$ be the moment of inertia of the body with respect to the principal axes, and define the body angular momentum $\bPi(t) \defeq \mathbb{I} \boldsymbol{\Omega}(t)$.
Let $\mathbf{e}_{3}$ be the unit vector pointing upward in the spatial frame and set $\bGamma(t) \defeq R(t)^{-1}\mathbf{e}_{3}$; this vector signifies the vertically upward direction~-- essentially the direction of gravity~-- seen from the body frame.
Then, the equations of motion are given by
\begin{gather} \label{eq:heavy_top}
 \dot{\bPi} = \bPi \times \big(\mathbb{I}^{-1}\bPi\big) + m g l \bGamma \times \mathbf{c}, \qquad \dot{\bGamma} = \bGamma \times \big(\mathbb{I}^{-1}\bPi\big),
\end{gather}
where $m$ is the mass of the body, $\mathbf{c} \in \R^{3}$ is the unit vector from the point of support to the direction of the body's center of mass in the body frame, and $l$ is the length of the line segment between these two points. The above set of equations is often called the \textit{heavy top equations}.

\subsection{Hamiltonian formulation and semidirect products}
The heavy top equations are known to be a Hamiltonian system. Let us define the \textit{heavy top bracket} for the space of smooth functions on the space $\R^{3} \times \R^{3} = \setdef{ (\bPi, \bGamma) }{ \bPi \in \R^{3}, \bGamma \in \R^{3} }$ as follows
\begin{gather} \label{eq:PB-heavy_top}
 \PB{f}{h}(\bPi, \bGamma) \defeq - \bPi \cdot \parentheses{ \pd{f}{\bPi} \times \pd{h}{\bPi} }
 - \bGamma \cdot \parentheses{ \pd{f}{\bPi} \times \pd{h}{\bGamma} - \pd{h}{\bPi} \times \pd{f}{\bGamma} }.
\end{gather}
The Hamiltonian $h\colon \R^{3} \times \R^{3} \to \R$ of the heavy top is given by
\begin{gather} \label{eq:h}
 h(\bPi, \bGamma) \defeq \frac{1}{2} \bPi \cdot \big(\mathbb{I}^{-1}\bPi\big) + m g l \bGamma \cdot \mathbf{c}.
\end{gather}
Then the Hamiltonian system defined with respect to the above Poisson bracket and this Hamiltonian, i.e.,
\begin{gather*}
 \dot{\bPi} = \PB{\bPi}{h}, \qquad \dot{\bGamma} = \PB{\bGamma}{h}
\end{gather*}
gives the heavy top equations~\eqref{eq:heavy_top}.

The heavy top bracket~\eqref{eq:PB-heavy_top} turns out to be the $(-)$-Lie--Poisson bracket on the dual $\se(3)^{*}$ of the Lie algebra $\se(3)$ of the special Euclidean group $\SE(3) \defeq \SO(3) \ltimes \R^{3}$; note that we identify $\se(3) = \so(3) \ltimes \R^{3}$~-- and hence $\se(3)^{*}$ as well~-- with $\R^{3} \times \R^{3}$ in the standard manner (see Appendix~\ref{sec:su2} for details). In fact, \cite{MaRaWe1984b, MaRaWe1984a} showed that the heavy top bracket is a special case of the semidirect product theory for reduction of Hamiltonian systems with broken symmetry; see also~\cite{HoMaRa1998a} for its Lagrangian counterpart. Specifically, the heavy top is originally a Hamiltonian system on $T^{*}\SO(3)$ where the presence of gravity breaks the $\SO(3)$-symmetry the system would otherwise possess. As a result, the system retains only the $\SO(2)$-symmetry with respect to rotations about the vertical axis. The semidirect product theory effectively recovers the full symmetry by considering the extended configuration space $\SE(3)$ as opposed to $\SO(3)$. As a~result, one obtains a Hamiltonian dynamics in $\se(3)^{*}$ defined in terms of the Hamiltonian~\eqref{eq:h} and the $(-)$-Lie--Poisson bracket~\eqref{eq:PB-heavy_top} on $\se(3)^{*} \cong \R^{3} \times \R^{3}$.

\subsection{Collective dynamics}
Given a Poisson manifold $P$ and an equivariant momentum map $\mathbf{M}\colon P \to \mathfrak{g}^{*}$ associated with a~right action of a Lie group $\mathsf{G}$ on $P$, one can show that $\mathbf{M}$ is a Poisson map with respect to the Poisson bracket on $P$ and the $(-)$-Lie--Poisson bracket on $\mathfrak{g}^{*}$; see, e.g., \cite[Theorem~12.4.1]{MaRa1999}.

Now, given a Hamiltonian system on $P$ with Hamiltonian $H\colon P \to \R$, suppose that there exists a function $h\colon \mathfrak{g}^{*} \to \R$ such that $h \circ \mathbf{M} = H$; such a function $h$ is called a \textit{collective Hamiltonian}. If $t \mapsto z(t)$ is a solution of the Hamiltonian system on $P$ defined by $H$, then \mbox{$t \mapsto \mathbf{M}(z(t))$} is a solution of the Lie--Poisson equation on $\mathfrak{g}^{*}$ defined in terms of the $(-)$-Lie--Poisson bracket and the collective Hamiltonian~$h$. This is the basic idea of \textit{collective motions/dynamics} originally due to \cite{GuSt1980}. Their motivating examples are aggregate motions of a number of particles ``as if it were a rigid body or liquid drop'', such as the liquid drop model in nuclear physics; hence the term ``collective''. See also \cite{HoMa1983} and \cite[Section~28]{GuSt1990} for details.

More recently, \cite{McMoVe2014,McMoVe2015} used this idea to develop geometric integrators for Lie--Poisson equations.
In one of their examples, they considered a natural right action of $\SU(2)$ on $\C^{2} \cong T^{*}\R^{2}$ and used its associated momentum map $\mathbf{M}\colon \C^{2} \to \su(2)^{*}$ to construct a natural Poisson map from $\C^{2} \cong T^{*}\R^{2}$ to $\so(3)^{*}$ where $\su(2)^{*} \cong \su(2)$ and $\so(3)^{*} \cong \so(3)$ are identified in a natural manner, as summarized in Appendix~\ref{sec:su2}.
As a result, they obtained a canonical Hamiltonian system on $T^{*}\R^{2}$ whose solutions are mapped by $\mathbf{M}$ to solutions of the free rigid body equation in $\so(3)^{*}$. By applying a symplectic integrator to the canonical Hamiltonian system on $T^{*}\R^{2}$ and taking the image by $\mathbf{M}$ of its numerical solutions, they obtained a Lie--Poisson integrator for the free rigid body equation and also for the point vortices equations on $\mathbb{S}^{2}$~-- a~coadjoint orbit in $\su(2)^{*}$.

We also note that \cite{Bogfjellmo-Collective} extended the above work to Hamiltonian systems on the direct product $\big(\mathbb{S}^{2}\big)^{n} \times T^{*}\R^{m}$ motivated by an application to the spin-lattice-electron equations. However, the symplectic/Poisson structure there is the standard one on the direct product (interactions between the $\big(\mathbb{S}^{2}\big)^{n}$ part and the $T^{*}\R^{m}$ part come from the Hamiltonian), whereas our Poisson structure is on a semidirect product and hence is geometrically more involved.

\subsection{Outline and main results}
We present a collective formulation of the heavy top equations by constructing a Poisson map from $T^{*}\C^{2} \cong T^{*}\R^{4}$ to $\se(3)^{*}$.
More specifically, we first consider, in Section~\ref{sec:J}, the cotangent lift of the natural right action of the semidirect product $\SU(2) \ltimes \C^{2}$ on $\C^{2} \cong \R^{4}$ and find the associated momentum map $\mathbf{L}\colon T^{*}\C^{2} \to \big(\su(2) \ltimes \C^{2}\big)^{*}$.

Next, in Section~\ref{sec:varpi}, we construct a Poisson map $\varpi\colon \big(\su(2) \ltimes \C^{2}\big)^{*} \to \se(3)^{*}$ with respect to the $(-)$-Lie--Poisson brackets on both duals. Then it follows that the composition $\mathbf{M}\colon T^{*}\C^{2} \to \se(3)^{*}$ defined as $\mathbf{M} \defeq \varpi \circ \mathbf{L}$ yields a Poisson map that gives a collective formulation of the heavy top dynamics. We note that a more natural and straightforward Poisson map $T^{*}\SE(3) \to \se(3)^{*}$ is \textit{not} appropriate for collective dynamics; see Remark~\ref{rem:why_not_SE3}.

In Section~\ref{sec:dynamics}, we state our main result, Theorem~\ref{thm:collective_dynamics}. It essentially states that we can realize any heavy top dynamics as a collective dynamics of a canonical Hamiltonian system in $T^{*}\C^{2} \cong T^{*}\R^{4}$. We also look into the symmetries possessed by the system in the general case as well as for the Lagrange top. It turns out that those momentum maps associated with these symmetries are essentially the well-known Casimirs $\norm{\bGamma}^{2}$ and $\bPi\cdot\bGamma$ of the general heavy top as well as the additional conserved quantity~$\Pi_{3}$ for the Lagrange top.

In Section~\ref{sec:integrator}, we develop a collective Lie--Poisson integrator for the heavy top dynamics by combining our result and the idea of \cite{McMoVe2014,McMoVe2015}.
Our test case is the Kovalevskaya top, which also possesses an additional conserved quantity~-- the Kovalevskaya invariant~-- in addition to the Casimirs. We apply both symplectic (implicit midpoint) and non-symplectic (explicit midpoint) integrators to the collective dynamics. We refer to the former as the \textit{collective Lie--Poisson integrator}. We also apply these integrators directly to the heavy top equations~\eqref{eq:heavy_top}; note that neither is a Poisson integrator for the heavy top equations although the implicit midpoint rule has favorable properties; see~\cite{AuKrWa1993}. We numerically demonstrate that the non-symplectic integrator, applied to either the collective or direct formulation, exhibits drifts in all the conserved quantities~-- the Hamiltonian, the Casimirs, and the Kovalevskaya invariant. On the other hand, the collective Lie--Poisson integrator preserves the Casimirs exactly (or more practically to machine precision) and exhibits near-conservation~-- small oscillations without drifts~-- of the exact values of the Hamiltonian and the Kovalevskaya invariant.

\section[Momentum map $\mathbf{L}\colon T^{*}\C^{2} \to \big(\su(2) \ltimes \C^{2}\big)^{*}$]{Momentum map $\boldsymbol{\mathbf{L}\colon T^{*}\C^{2} \to \big(\su(2) \ltimes \C^{2}\big)^{*}}$}\label{sec:J}
\subsection[Cotangent bundles $T^{*}\C^{2}$ and $T^{*}\R^{4}$]{Cotangent bundles $\boldsymbol{T^{*}\C^{2}}$ and $\boldsymbol{T^{*}\R^{4}}$}
Let $\chi \in \C^{2}$ be arbitrary and identify it with $q \in \R^{4}$ as follows
\begin{gather*}
 \C^{2} \to \R^{4},
 \qquad
 \chi = \begin{bmatrix}
 \chi_{1} \\
 \chi_{2}
 \end{bmatrix} =
 \begin{bmatrix}
 q_{1} + \rmi q_{2} \\
 q_{3} + \rmi q_{4}
 \end{bmatrix}
 \mapsto q = (q_{1}, q_{2}, q_{3}, q_{4}).
\end{gather*}
We equip $\C^{2}$ with the inner product
\begin{gather}
 \label{eq:ip-C2}
 \ip{\,\cdot\,}{\,\cdot\,}_{\C^{2}}\colon \ \C^{2} \times \C^{2} \to \R, \qquad (\psi,\phi) \mapsto \ip{\psi}{\phi} = \Re(\psi^{*}\phi).
\end{gather}
Then it is easy to see that this inner product is compatible with the dot product in $\R^{4}$.
Setting
\begin{gather*}
 \psi =
 \begin{bmatrix}
 \psi_{1} \\
 \psi_{2}
 \end{bmatrix}
 =
 \begin{bmatrix}
 p_{1} + \rmi p_{2} \\
 p_{3} + \rmi p_{4}
 \end{bmatrix}
 \leftrightarrow p = (p_{1}, p_{2}, p_{3}, p_{4}),
\end{gather*}
we have
\begin{gather*}
 \ip{\psi}{\chi}_{\C^{2}} = p \cdot q.
\end{gather*}
Therefore, we may identify the dual $\big(\C^{2}\big)^{*} \cong \big(\R^{4}\big)^{*}$ with $\C^{2} \cong \R^{4}$ via this inner product.
As a~result, we may identify the cotangent bundles $T^{*}\C^{2}$ and $T^{*}\R^{4}$ as follows
\begin{gather*}
 T^{*}\C^{2} \to T^{*}\R^{4},
 \qquad
 (\chi, \psi) =
 \parentheses{
 \begin{bmatrix}
 q_{1} + \rmi q_{2} \\
 q_{3} + \rmi q_{4}
 \end{bmatrix},
 \begin{bmatrix}
 p_{1} + \rmi p_{2} \\
 p_{3} + \rmi p_{4}
 \end{bmatrix}
 }
 \mapsto (q, p) = (q_{1}, \dots, q_{4}, p_{1}, \dots, p_{4}).
\end{gather*}
Then, for any $\chi \in \C^{2}$, the natural dual pairing between $T^{*}_{\chi}\C^{2} \cong \big(\C^{2}\big)^{*}$ and $T_{\chi}\C^{2} \cong \C^{2}$ is identical to the inner product~\eqref{eq:ip-C2}, i.e., with an abuse of notation,
\begin{gather*}
 \ip{\,\cdot\,}{\,\cdot\,}_{\C^{2}}\colon \ T^{*}_{\chi}\C^{2} \times T_{\chi}\C^{2} \to \R, \qquad (\psi, \dot{\chi}) \mapsto \Re(\psi^{*}\dot{\chi}).
\end{gather*}

We define a canonical 1-form on $T^{*}\C^{2}$ as
\begin{gather*}
 \Theta_{1} \defeq \Re(\psi^{*}\d\chi) = \Re\big( \bar{\psi}_{i}\d\chi_{i}\big),
\end{gather*}
and define a symplectic form on $T^{*}\C^{2}$ as follows
\begin{gather*} 
 \Omega \defeq -\d\Theta_{1} = \Re\big(\d\chi_{i} \wedge \d\bar{\psi}_{i}\big).
\end{gather*}
Note that we use Einstein's summation convention unless otherwise stated. One easily sees that, under the above identification of $T^{*}\C^{2}$ and $T^{*}\R^{4}$, $\Theta_{1}$ and $\Omega$ become the canonical 1-form $p_{i}\d{q}_{i}$ and the canonical symplectic form $\d{q}_{i} \wedge \d{p}_{i}$ on $T^{*}\R^{4}$, respectively. Therefore, $T^{*}\C^{2}$ is also equipped with the following canonical Poisson bracket. For any smooth \smash{$F, H \colon T^{*}\C^{2} \!\cong\! T^{*}\R^{4} \to \R$},
\begin{gather} \label{eq:PB-TstarC2}
 \PB{F}{H}_{T^{*}\C^{2}}(q,p) \defeq \pd{F}{q_{i}} \pd{H}{p_{i}} - \pd{F}{p_{i}} \pd{H}{q_{i}}.
\end{gather}

Alternatively, we may define the one-form
\begin{gather} \label{eq:Theta_2}
 \Theta_{2} \defeq \frac{1}{2}\Re(\psi^{*}\d\chi - \chi^{*}\d\psi).
\end{gather}
Then, it is $(p_{i}\d{q}_{i} - q_{i}\d{p}_{i})/2$ in $T^{*}\R^{4}$, and satisfies $\Omega = -\d\Theta_{2}$ as well.

\subsection[Actions of $\SU(2) \ltimes \C^{2}$]{Actions of $\boldsymbol{\SU(2) \ltimes \C^{2}}$}
Let us define the semidirect product $\SU(2) \ltimes \C^{2}$ using the natural action of $\SU(2)$ on $\C^{2}$, i.e., for any $(U, \rho), (W, \vartheta) \in \SU(2) \times \C^{2}$, we define a binary operation
\begin{gather*}
 (U, \rho) \cdot (W, \vartheta) \defeq (U W, U \vartheta + \rho).
\end{gather*}
This renders $\SU(2) \ltimes \C^{2}$ a Lie group. Alternatively, one may think of the resulting Lie group $\SU(2) \ltimes \C^{2}$ as the matrix group{\samepage
\begin{gather*}
 \setdef{ \begin{bmatrix}
 U & \rho \\
 0 & 1
 \end{bmatrix}
 \in \mathsf{GL}(3,\C) }{ U \in \SU(2),\, \rho \in \C^{2} }
\end{gather*}
under the standard matrix multiplication.}

Now consider the (right) action of $\SU(2) \ltimes \C^{2}$ on $\C^{2}$ defined as follows
\begin{gather*}
 \Psi\colon \ \big(\SU(2) \ltimes \C^{2}\big) \times \C^{2} \to \C^{2}, \qquad ((U, \rho), \chi) \mapsto \Psi_{(U,\rho)}(\chi) \defeq U^{*}(\chi - \rho).
\end{gather*}
In terms of matrices, one may write this action as follows
\begin{gather*}
 \parentheses{ \begin{bmatrix}
 U & \rho \\
 0 & 1
 \end{bmatrix},
 \begin{bmatrix}
 \chi \\
 1
 \end{bmatrix} }
 \mapsto \begin{bmatrix}
 U & \rho \\
 0 & 1
 \end{bmatrix}^{-1}
 \begin{bmatrix}
 \chi \\
 1
 \end{bmatrix} =
 \begin{bmatrix}
 U^{*} & -U^{*}\rho \\
 0 & 1
 \end{bmatrix}
 \begin{bmatrix}
 \chi \\
 1
 \end{bmatrix} =
 \begin{bmatrix}
 U^{*} (\chi - \rho) \\
 1
 \end{bmatrix}.
\end{gather*}
Its cotangent lift is then
\begin{align*}
 T^{*}\Psi\colon \ & \big(\SU(2) \ltimes \C^{2}\big) \times T^{*}\C^{2} \to T^{*}\C^{2}, \\
 & ((U, \rho), (\chi, \psi)) \mapsto T^{*}\Psi_{(U,\rho)^{-1}}(\chi,\psi) \defeq (U^{*}(\chi - \rho), U^{*}\psi).
\end{align*}
It is clear that the canonical 1-form $\Theta_{1}$ is invariant under the cotangent lift, i.e., $(T^{*}\Psi_{(U,\rho)^{-1}})^{*}\Theta_{1} \allowbreak = \Theta_{1}$ for any $(U, \rho) \in \SU(2) \ltimes \C^{2}$.

\subsection{Momentum map}
Let us find the momentum map associated with the above action. Let $(\xi, \phi)$ be an arbitrary element in the Lie algebra $\su(2) \ltimes \C^{2} = T_{(I,0)}\big(\SU(2) \ltimes \C^{2}\big)$. Its infinitesimal generator on $\C^{2}$ is defined as
\begin{gather*}
 (\xi, \phi)_{\C^{2}}(\chi) \defeq \dzero{s}{\Psi_{\exp(s(\xi,\phi))}(\chi)},
\end{gather*}
where the exponential map takes the form
\begin{gather*}
 \exp(s(\xi,\phi)) = \parentheses{ \exp(s\xi), \int_{0}^{s} \exp(\sigma\xi)\phi\,{\rm d}\sigma }.
\end{gather*}
Therefore, we obtain
\begin{gather*}
 (\xi, \phi)_{\C^{2}}(\chi) = \dzero{s}{\exp(-s\xi)\parentheses{ \chi - \int_{0}^{s} \exp(\sigma\xi)\phi\,{\rm d}\sigma }}
= -\xi\chi - \phi.
\end{gather*}
As a result, the associated momentum map $\mathbf{L}\colon T^{*}\C^{2} \to \big(\su(2) \ltimes \C^{2}\big)^{*}$ satisfies
\begin{align*}
 \ip{ \mathbf{L}(\chi, \psi) }{ (\xi, \phi) }
 &= \Theta_{1}(\chi, \psi) \cdot (\xi, \phi)_{\C^{2}}(\chi, \psi) = -\Re\parentheses{ \psi^{*}(\xi\chi + \phi) } \\
 &= -\Re(\tr(\chi\psi^{*}\xi)) - \Re(\psi^{*} \phi) \\
 &= -\frac{1}{2} \parentheses{ \tr(\chi\psi^{*}\xi) + \tr(\xi^{*}\psi\chi^{*}) } - \ip{\psi}{\phi}_{\C^{2}} \\
 &= 2\tr\parentheses{ \frac{1}{4}\parentheses{ \psi\chi^{*} - \chi\psi^{*} } \xi } + \ip{-\psi}{\phi}_{\C^{2}} \\
 &= 2\tr\parentheses{ \frac{1}{4}\parentheses{ \chi\psi^{*} - \psi\chi^{*} - \rmi\Im(\psi^{*}\chi)I }^{*} \xi } + \ip{-\psi}{\phi}_{\C^{2}} \\
 &= \ip{ \frac{1}{4}\parentheses{ \chi\psi^{*} - \psi\chi^{*} - \rmi\Im(\psi^{*}\chi)I } }{ \xi }_{\su(2)} + \ip{-\psi}{\phi}_{\C^{2}},
\end{align*}
where we used the inner products on $\su(2)$ and $\C^{2}$ defined in \eqref{eq:ip-su2} and \eqref{eq:ip-C2}, respectively.
Note also that, in the second last line, we added an extra term to render the matrix inside the conjugate transpose traceless to come up with an element in $\su(2)$.
Using the inner products, we may identify $\su(2)^{*}$ with $\su(2)$ and $\big(\C^{2}\big)^{*}$ with $\C^{2}$ so that $\big(\su(2) \ltimes \C^{2}\big)^{*} \cong \su(2) \times \C^{2}$. Under this identification, we obtain
\begin{gather*}
 \mathbf{L}(\chi, \psi) = \parentheses{ \frac{1}{4}\parentheses{ \chi\psi^{*} - \psi\chi^{*} - \rmi\Im(\psi^{*}\chi)I },\, -\psi }.
\end{gather*}

Since this momentum map is associated with the cotangent lift action of the right action $\Psi$ on $T^{*}\C^{2}$, the momentum map $\mathbf{L}$ is equivariant (see \cite[Theorem~12.1.4]{MaRa1999}), i.e., for any $(U,\rho) \in \SU(2) \ltimes \C^{2}$,
\begin{gather*}
 \mathbf{L} \circ T^{*}\Psi_{(U,\rho)^{-1}} = \Ad_{(U,\rho)}^{*} \circ \mathbf{L}.
\end{gather*}

\begin{Remark}[why not $\SE(3)$ action on $\R^{3}$?] \label{rem:why_not_SE3} Perhaps the most natural way to construct an $\se(3)^{*}$-valued equivariant momentum map~-- under a right action~-- on a cotangent bundle would be the following. Consider the right $\SE(3)$-action on $\R^{3}$ defined as
 \begin{gather*}
 \SE(3) \times \R^{3} \to \R^{3}, \qquad ((R,\mathbf{a}),\mathbf{x}) \mapsto R^{T}(\mathbf{x} - \mathbf{a}),
 \end{gather*}
 or using matrices,
 \begin{gather*}
 \parentheses{
 \begin{bmatrix}
 R & \mathbf{a} \\
 0 & 1
 \end{bmatrix},
 \begin{bmatrix}
 \mathbf{x} \\
 1
 \end{bmatrix}
 }
 \mapsto \begin{bmatrix}
 R & \mathbf{a} \\
 0 & 1
 \end{bmatrix}^{-1}
 \begin{bmatrix}
 \mathbf{x} \\
 1
 \end{bmatrix} =
 \begin{bmatrix}
 R^{T}(\mathbf{x} - \mathbf{a}) \\
 1
 \end{bmatrix},
 \end{gather*}
 that is, this is the right action one obtains from the natural left action of $\SE(3)$ on $\R^{3}$. Using its cotangent lift action on $T^{*}\R^{3}$, one obtains the momentum map
 \begin{gather*}
 \tilde{\mathbf{M}}\colon \ T^{*}\R^{3} \to \se(3)^{*}, \qquad
 \tilde{\mathbf{M}}(\mathbf{q}, \mathbf{p}) = -(\mathbf{q} \times \mathbf{p}, \mathbf{p})
 \end{gather*}
 with the standard identification $\se(3)^{*} \cong \so(3) \times \R^{3} \cong \R^{3} \times \R^{3}$.
 However, this momentum map is too restrictive for the collective heavy top dynamics in $\se(3)^{*}$ because setting $(\bPi, \bGamma) = \tilde{\mathbf{M}}(\mathbf{q}, \mathbf{p})$ implies that the angular momentum $\bPi$ is always perpendicular to $\bGamma$.
 This is clearly not always the case.
 For example, even the very simple case of the top spinning in the upright position does not satisfy this condition: $\bPi$ and $\bGamma$ are both parallel to $(0,0,1)$.
\end{Remark}

Let us show that our momentum map does not suffer from the restriction mentioned in the above remark.
\begin{Proposition} \label{prop:image_of_J}
 Under the identification $\big(\su(2) \ltimes \C^{2}\big)^{*} \cong \su(2) \times \C^{2}$, the image of $\mathbf{L}$ contains $\su(2) \times \big(\C^{2}\backslash\{0\}\big)$, i.e., $\su(2) \times \big(\C^{2}\backslash\{0\}\big) \subset \mathbf{L}\big(T^{*}\C^{2}\big)$.
\end{Proposition}

\begin{Remark} \label{rem:image_of_J} As we shall see later in Section~\ref{sec:varpi}, we will construct a map $\varpi\colon \big(\su(2) \ltimes \C^{2}\big)^{*} \to \se(3)^{*}$, and consider the composition $\mathbf{M} \defeq \varpi \circ \mathbf{L}\colon T^{*}\C^{2} \to \se(3)^{*}$. It then turns out that the origin $\psi = 0$ in $\C^{2}$ corresponds to the origin $\bGamma = 0$ in $\R^{3}$ (see~\eqref{eq:L-expression} below) and hence has no practical importance in the heavy top dynamics because $\bGamma$ is a unit vector.
 This proposition makes sure that the image of $\mathbf{M}$ contains any possible values of $\bPi \in \R^{3}$ and $\bGamma \in \R^{3}\backslash\{0\}$, and so setting $(\bPi, \bGamma) = \mathbf{M}(\chi, \psi)$ does not impose any practical restriction in the collective heavy top dynamics.
\end{Remark}

\begin{proof} Let $\psi \in \C^{2}\backslash\{0\}$ be arbitrary and consider the map
 \begin{gather*}
 \C^{2} \to \su(2)^{*}, \qquad \chi \mapsto \frac{1}{4}\big( \chi\psi^{*} - \psi\chi^{*} - \rmi\Im(\psi^{*}\chi)I\big).
 \end{gather*}
 Then it suffices to show that this map is surjective. In fact, one may rewrite the map with the identification $\C^{2} \cong \R^{4}$ defined by $(\chi,\psi) \mapsto (q,p)$ as shown above as well as the identification $\su(2)^{*} \cong \R^{3}$ (see Appendix~\ref{sec:su2} below) to obtain the following linear map
 \begin{gather*}
 \R^{4} \to \R^{3}, \qquad q \mapsto A q \qquad\text{with}\quad
 A \defeq
 \frac{1}{2}
 \begin{bmatrix}
 p_{4} & -p_{3} & p_{2} & -p_{1} \\
 -p_{3} & -p_{4} & p_{1} & p_{2} \\
 p_{2} & -p_{1} & -p_{4} & p_{3}
 \end{bmatrix}.
 \end{gather*}
 Then $A A^{T} = \frac{\norm{p}^{2}}{4} I$, whereas $p \neq 0$ because $\psi \neq 0$. Hence this linear map is surjective.
\end{proof}

\subsection[Lie--Poisson bracket on $\big(\su(2) \ltimes \C^{2}\big)^{*}$]{Lie--Poisson bracket on $\boldsymbol{\big(\su(2) \ltimes \C^{2}\big)^{*}}$}
Let us equip $\big(\su(2) \ltimes \C^{2}\big)^{*}$ with the $(-)$-Lie--Poisson bracket. For any smooth $f, h\colon \big(\su(2) \ltimes \C^{2}\big)^{*} \to \R$, we define
\begin{gather} \label{eq:PB-gstar}
 \PB{f}{h}_{(\su(2) \ltimes \C^{2})^{*}}(\mu, \alpha) \defeq -\ip{\mu}{ \left[ \fd{f}{\mu}, \fd{h}{\mu} \right] }_{\su(2)}
 - \ip{\alpha}{ \fd{f}{\mu} \fd{h}{\alpha} - \fd{h}{\mu} \fd{f}{\alpha} }_{\C^{2}},
\end{gather}
where $(\tfd{f}{\mu}, \tfd{f}{\alpha}) \in \su(2) \times \C^{2}$ (evaluated at $(\mu,\alpha)$) is defined so that, for any $(\delta\mu, \delta\alpha) \in \big(\su(2) \ltimes \C^{2}\big)^{*}$,
\begin{gather} \label{eq:fd-gstar}
 \ip{ \delta\mu }{ \fd{f}{\mu} }_{\su(2)} + \ip{ \delta\alpha }{ \fd{f}{\alpha} }_{\C^{2}} = \dzero{s}{ f(\mu + s\delta\mu, \alpha + s\delta\alpha) }.
\end{gather}
Then the equivariance of $\mathbf{L}$ implies that it is Poisson with respect to the canonical Poisson bracket~\eqref{eq:PB-TstarC2} and the above Lie--Poisson bracket~\eqref{eq:PB-gstar} (see, e.g., \cite{GuSt1980}, \cite[Section~28]{GuSt1990}, and \cite[Theorem~12.4.1]{MaRa1999}):
\begin{gather*}
 \PB{f \circ \mathbf{L}}{h \circ \mathbf{L}}_{T^{*}\C^{2}} = \PB{f}{h}_{(\su(2) \ltimes \C^{2})^{*}} \circ \mathbf{L}.
\end{gather*}

\section[Poisson map $\varpi\colon \big(\su(2) \ltimes \C^{2}\big)^{*} \to \se(3)^{*}$]{Poisson map $\boldsymbol{\varpi\colon \big(\su(2) \ltimes \C^{2}\big)^{*} \to \se(3)^{*}}$}\label{sec:varpi}
\subsection[From $\big(\su(2) \ltimes \C^{2}\big)^{*}$ to $\se(3)^{*}$]{From $\boldsymbol{\big(\su(2) \ltimes \C^{2}\big)^{*}}$ to $\boldsymbol{\se(3)^{*}}$}
Just like we identified $\big(\su(2) \ltimes \C^{2}\big)^{*}$ with $\su(2) \times \C^{2}$, we also identify $\se(3)^{*}$ with $\so(3) \times \R^{3}$ via the inner product~\eqref{eq:ip-so3} on $\so(3)$ and the dot product on $\R^{3}$.

Under this identification, let us define
\begin{subequations} \label{eq:varpi}
 \begin{gather}
 \varpi\colon \ \big(\su(2) \ltimes \C^{2}\big)^{*} \to \se(3)^{*}, \qquad (\mu, \alpha) \mapsto
 \parentheses{ \hat{\mu}, \varpi_{2}(\alpha) }
 \end{gather}
 with
 \begin{gather} \label{eq:varpi2}
 \varpi_{2}\colon \ \C^{2} \to \R^{3}, \qquad
 \alpha \mapsto \big( 2\Re(\bar{\alpha}_{1}\alpha_{2}),\, 2\Im(\bar{\alpha}_{1}\alpha_{2}),\, |\alpha_{1}|^{2} - |\alpha_{2}|^{2}\big).
 \end{gather}
\end{subequations}
The first part of the map is the well-known identification of $\su(2)$ with $\so(3)$ summarized in Appendix~\ref{sec:su2} below. The second part of the map is also well known in the context of the Hopf fibration. Particularly, $\varpi_{2}$ maps the three-sphere with radius $\sqrt{R}$ centered at the origin in $\C^{2}$, i.e.,
\begin{gather*}
 \mathbb{S}^{3}_{\sqrt{R}} \defeq \big\{ \alpha \in \C^{2} \,|\, \norm{\alpha} = \sqrt{R} \big\},
\end{gather*}
to the two-sphere with radius $R$ centered at the origin in $\R^{3}$.

Now let us define
\begin{gather} \label{eq:L}
 \mathbf{M}\colon \ T^{*}\C^{2} \to \se(3)^{*}, \qquad \mathbf{M} \defeq \varpi \circ \mathbf{L},
\end{gather}
i.e., so that the diagram
\begin{equation*}
 \begin{tikzcd}[column sep=8ex, row sep=8ex]
 T^{*}\C^{2} \arrow{r}{\mathbf{L}} \arrow[swap]{dr}{\mathbf{M}} & (\su(2) \ltimes \C^{2})^{*} \arrow{d}{\varpi} \\
 & \se(3)^{*}
 \end{tikzcd}
\end{equation*}
commutes.

\subsection[$\varpi$ is a Poisson map]{$\boldsymbol{\varpi}$ is a Poisson map}
Let us equip $\se(3)^{*}$ with the $(-)$-Lie--Poisson bracket as well. For any smooth $f, h\colon \se(3)^{*} \to \R$, we define
\begin{gather} \label{eq:PB-se3star}
 \PB{f}{h}_{\se(3)^{*}}(\hat{\Pi}, \bGamma) \defeq
 -\ip{\hat{\Pi}}{ \left[ \fd{f}{\hat{\Pi}}, \fd{h}{\hat{\Pi}} \right] }_{\so(3)}
 - \bGamma \cdot \parentheses{ \fd{f}{\hat{\Pi}} \pd{h}{\bGamma} - \fd{h}{\hat{\Pi}} \pd{f}{\bGamma} },
\end{gather}
where $\tfd{f}{\hat{\Pi}} \in \so(3)$ (evaluated at $(\hat{\Pi}, \bGamma)$) is defined in a similar manner as in \eqref{eq:fd-gstar}.
With the identification $\so(3) \cong \R^{3}$, this Poisson bracket is nothing but the heavy top bracket~\eqref{eq:PB-heavy_top}.
Then we have the following:

\begin{Proposition} \label{prop:varpi}
 The map $\varpi\colon \big(\su(2) \ltimes \C^{2}\big)^{*} \to \se(3)^{*}$ is Poisson with respect to the $(-)$-Lie--Poisson brackets \eqref{eq:PB-gstar} and \eqref{eq:PB-se3star}, i.e., for any smooth $f, h\colon \se(3)^{*} \to \R$,
 \begin{gather*}
 \PB{f \circ \varpi}{h \circ \varpi}_{(\su(2) \ltimes \C^{2})^{*}} = \PB{f}{h}_{\se(3)^{*}} \circ \varpi.
 \end{gather*}
\end{Proposition}

\begin{proof} See Appendix~\ref{sec:prop:varpi-proof}.
\end{proof}

\begin{Corollary} \label{cor:L} The map $\mathbf{M} \defeq \varpi\circ\mathbf{L} \colon T^{*}\C^{2} \to \se(3)^{*}$ is Poisson with respect to the Poisson brackets \eqref{eq:PB-TstarC2} and \eqref{eq:PB-se3star}, i.e., for any smooth $f, h\colon \se(3)^{*} \to \R$,
 \begin{gather*}
 \PB{f \circ \mathbf{M}}{h \circ \mathbf{M}}_{T^{*}\C^{2}} = \PB{f}{h}_{\se(3)^{*}} \circ \mathbf{M}.
 \end{gather*}
\end{Corollary}

\section{Collective heavy top dynamics}\label{sec:dynamics}
\subsection{Collective dynamics}
Concrete expressions of the map $\mathbf{M}$ defined above in \eqref{eq:L} are the following
\begin{align}
 \mathbf{M}(\chi, \psi) &= \parentheses{
 \frac{1}{2}
 \begin{bmatrix}
 -\Im\parentheses{ \chi_{1} \bar{\psi}_{2} + \chi_{2} \bar{\psi}_{1} } \vspace{1mm}\\
 \Re\parentheses{ \chi_{2} \bar{\psi}_{1} - \chi_{1} \bar{\psi}_{2} } \vspace{1mm}\\
 \Im\parentheses{ \chi_{2} \bar{\psi}_{2} - \chi_{1} \bar{\psi}_{1} }
 \end{bmatrix},
 \begin{bmatrix}
 2\Re(\bar{\psi}_{1}\psi_{2}) \vspace{1mm}\\
 2\Im(\bar{\psi}_{1}\psi_{2}) \vspace{1mm}\\
 |\psi_{1}|^{2} - |\psi_{2}|^{2}
 \end{bmatrix}
 } \nonumber\\
 &= \parentheses{
 \frac{1}{2}
 \begin{bmatrix}
 q_{1} p_{4} - q_{4} p_{1} - q_{2} p_{3} + q_{3} p_{2} \\
 q_{3} p_{1} - q_{1} p_{3} - q_{2} p_{4} + q_{4} p_{2} \\
 q_{1} p_{2} - q_{2} p_{1} - q_{3} p_{4} + q_{4} p_{3}
 \end{bmatrix},
 \begin{bmatrix}
 2(p_{1} p_{3} - p_{2} p_{4}) \\
 -2(p_{1} p_{4} + p_{2} p_{3}) \\
 p_{1}^{2} + p_{2}^{2} - p_{3}^{2} - p_{4}^{2}
 \end{bmatrix}
 }.\label{eq:L-expression}
\end{align}
Now define a Hamiltonian $H\colon T^{*}\C^{2} \cong T^{*}\R^{4} \to \R$ as
\begin{gather*}
 H \defeq h \circ \mathbf{M},
\end{gather*}
where $h\colon \se(3)^{*} \to \R$ is the heavy top Hamiltonian defined in~\eqref{eq:h}. Then we have
\begin{gather}
 H(\chi,\psi)= \frac{1}{8}\parentheses{
 \frac{ \parentheses{ \Im\parentheses{ \chi_{1} \bar{\psi}_{2} + \chi_{2} \bar{\psi}_{1} } }^{2} }{I_{1}}
 + \frac{ \parentheses{ \Re\parentheses{ \chi_{2} \bar{\psi}_{1} - \chi_{1} \bar{\psi}_{2} } }^{2} }{I_{2}}
 + \frac{ \parentheses{ \Im\parentheses{ \chi_{2} \bar{\psi}_{2} - \chi_{1} \bar{\psi}_{1} } }^{2} }{I_{3}}
 }\nonumber\\
\hphantom{H(\chi,\psi)=}{} + m g l \parentheses{2\Re(\bar{\psi}_{1}\psi_{2}) c_{1} + 2\Im(\bar{\psi}_{1}\psi_{2}) c_{2} + ( |\psi_{1}|^{2} - |\psi_{2}|^{2} ) c_{3}
 }. \label{eq:H}
\end{gather}
Define a Hamiltonian vector $X_{H} \in \mathfrak{X}(T^{*}\C^{2})$ by setting
\begin{gather} \label{eq:Hamiltonian_system-H}
 \ins{X_{H}}{\Omega} = \d{H},
\end{gather}
or equivalently, using the Poisson bracket~\eqref{eq:PB-TstarC2},
\begin{gather*}
 \dot{\chi} = \PB{ \chi}{ H }_{T^{*}\C^{2}}, \qquad \dot{\psi} = \PB{ \psi }{ H }_{T^{*}\C^{2}},
\end{gather*}
which yield
\begin{subequations} \label{eq:Hamiltonian_system-H-coordinates}
 \begin{gather}
 \dot{\chi} = 2 \pd{H}{\bar{\psi}}, \qquad \dot{\psi} = -2 \pd{H}{\bar{\chi}}.
 \end{gather}
 In terms of the real coordinate $(q,p)$ for $T^{*}\R^{4}$, the symplectic form $\Omega$ is canonical, and so we have the canonical Hamiltonian system
 \begin{gather}
 \dot{q} = \pd{H}{p}, \qquad \dot{p} = -\pd{H}{q}.
 \end{gather}
\end{subequations}

We are now ready to state our main result:
\begin{Theorem}[collective heavy top dynamics] \label{thm:collective_dynamics}\quad
 \begin{enumerate}\itemsep=0pt
 \item[$(i)$] For any $(\bPi, \bGamma) \in \R^{3} \times \big(\R^{3}\backslash\{0\}\big)$, there exists a corresponding element $(\chi, \psi) \in T^{*}\C^{2}$ such that $\mathbf{M}(\chi, \psi) = (\bPi, \bGamma)$.
 \item[$(ii)$] Let $I \subset \R$ be a time interval, and let $\Phi\colon I \times T^{*}\C^{2} \to T^{*}\C^{2}$ and $\varphi\colon I \times \se(3)^{*} \to \se(3)^{*}$ be the flows of the canonical Hamiltonian system~\eqref{eq:Hamiltonian_system-H-coordinates} and of the heavy top equations~\eqref{eq:heavy_top}, respectively.
 Then, for any $t \in I$,
 \begin{gather*}
 \varphi_{t} \circ \mathbf{M} = \mathbf{M} \circ \Phi_{t}.
 \end{gather*}
 \end{enumerate}
\end{Theorem}

\begin{proof} (i)~From Proposition~\ref{prop:image_of_J}, we have $\su(2) \times \big(\C^{2}\backslash\{0\}\big) \subset \mathbf{L}\big(T^{*}\C^{2}\big)$, and so
 \begin{gather*}
 \varpi\parentheses{ \su(2) \times \big(\C^{2}\backslash\{0\}\big) } \subset \mathbf{M}\big(T^{*}\C^{2}\big).
 \end{gather*}
 However, from the definition of $\varpi$ in \eqref{eq:varpi}, we have
 \begin{gather*}
 \varpi\parentheses{ \su(2) \times \big(\C^{2}\backslash\{0\}\big) } = \so(3) \times \big(\R^{3}\backslash\{0\}\big),
 \end{gather*}
 because the second part $\varpi_{2}$ maps $\mathbb{S}^{3}_{\sqrt{R}} \subset \C^{2}$ to $\mathbb{S}^{2}_{R} \subset \R^{3}$ for any $R > 0$; note also that we identified $\so(3)^{*}$ with~$\so(3)$ here. Therefore, the assertion follows upon the identification of~$\so(3)$ with~$\R^{3}$.

 (ii)~It follows easily from Corollary~\ref{cor:L}: The fact that $\mathbf{M}$ is Poisson implies that the map~$\mathbf{M}$ pushes the flow~$\Phi$ to $\varphi$ (see, e.g., \cite[Proposition~10.3.2]{MaRa1999}).
\end{proof}

\begin{Remark} The first assertion is the result alluded in Remark~\ref{rem:image_of_J}. It shows that any solution of the heavy top equations~\eqref{eq:heavy_top} can be realized as the image by $\mathbf{M}$ of a corresponding solution of the canonical Hamiltonian system~\eqref{eq:Hamiltonian_system-H-coordinates}.
\end{Remark}

\subsection{Symmetry and conserved quantities of general heavy top}\label{sec:symmetry_and_conservation_laws}
Suppose that $F\colon T^{*}\C^{2} \to \R$ is a conserved quantity of the (canonical) Hamiltonian system~\eqref{eq:Hamiltonian_system-H} or~\eqref{eq:Hamiltonian_system-H-coordinates}, and that~$F$ is collective, i.e., there exists $f\colon \se(3)^{*} \to \R$ such that $f \circ \mathbf{M} = F$.
Then one easily sees that~$f$ is a conserved quantity of the heavy top dynamics.

There are two conserved quantities of the Hamiltonian system~\eqref{eq:Hamiltonian_system-H} associated with its symmetries. For the first one, consider the following $\R$-action on $T^{*}\C^{2}$:
\begin{gather*}
 \R \times T^{*}\C^{2} \to T^{*}\C^{2}, \qquad (b, (\chi, \psi)) \mapsto (\chi + b\psi, \psi).
\end{gather*}
We easily see that the alternative one-form $\Theta_{2}$ defined in \eqref{eq:Theta_2} is invariant under this action, and hence so is $\Omega$, i.e., this action is symplectic.
It is also easy to see that the Hamiltonian~\eqref{eq:H} is invariant under this action as well, i.e., $H(\chi + b\psi, \psi) = H(\chi, \psi)$ for any $b \in \R$.
Let $v \in \R$ be an the element of the Lie algebra $T_{0}\R \cong \R$ of $\R$. Its infinitesimal generator is then
\begin{gather*}
 v_{T^{*}\C^{2}}(\chi, \psi) = v \parentheses{ \psi_{i} \pd{}{\chi_{i}} + \bar{\psi}_{i} \pd{}{\bar{\chi}_{i}} } = v\,p_{j}\pd{}{q_{j}}.
\end{gather*}
Let $J_{1}\colon T^{*}\C^{2} \to \R$ be the associated momentum map. Then, it satisfies
\begin{gather*}
 J_{1}(\chi, \psi) \cdot v= \Theta_{1}(\chi, \psi) \cdot v_{T^{*}\C^{2}}(\chi, \psi) = \frac{1}{2} \Re\big(v \norm{\psi}^{2}\big) = \frac{1}{2} \norm{\psi}^{2} \cdot v.
\end{gather*}
Hence we obtain $J_{1}(\chi, \psi) = \norm{\psi}^{2}/2$. Therefore,
\begin{gather*}
 F_{1}\colon \ T^{*}\C^{2} \to \R, \qquad F_{1}(\chi,\psi) \defeq 4 J_{1}(\chi,\psi)^{2} = \norm{\psi}^{4}
\end{gather*}
is a conserved quantity of the Hamiltonian system~\eqref{eq:Hamiltonian_system-H} by Noether's theorem; see, e.g., \cite[Theorem~11.4.1, p.~372]{MaRa1999}. One easily sees that $f_{1} \circ \mathbf{M} = F_{1}$ with
\begin{gather*}
 f_{1}\colon \ \se(3)^{*} \to \R, \qquad f_{1}(\bPi,\bGamma) \defeq \norm{\bGamma}^{2}.
\end{gather*}
In fact, this is a well-known Casimir of the heavy top bracket~\eqref{eq:PB-heavy_top} or \eqref{eq:PB-se3star}.

For the second one, consider the following $\SO(2) \cong \mathbb{S}^{1}$-action on $\C^{2}$:
\begin{gather*}
 \mathbb{S}^{1} \times \C^{2} \to \C^{2}, \qquad \big({\rm e}^{\rmi\theta}, \chi\big) \mapsto {\rm e}^{\rmi\theta}\chi.
\end{gather*}
Its cotangent lift is
\begin{gather*}
 \mathbb{S}^{1} \times T^{*}\C^{2} \to T^{*}\C^{2}, \qquad \big({\rm e}^{\rmi\theta}, (\chi,\psi)\big) \mapsto \big({\rm e}^{\rmi\theta}\chi, {\rm e}^{\rmi\theta}\psi\big),
\end{gather*}
and the Hamiltonian \eqref{eq:H} is invariant under this action, i.e., $H\big({\rm e}^{\rmi\theta}\chi, {\rm e}^{\rmi\theta}\psi\big) = H(\chi, \psi)$ for any ${\rm e}^{\rmi\theta} \in \mathbb{S}^{1}$. Its associated momentum map is
\begin{gather*}
 J_{2}\colon \ T^{*}\C^{2} \to \so(2)^{*}\cong \R, \qquad J_{2}(\chi,\psi) \defeq -\Im(\psi^{*}\chi),
\end{gather*}
and again by Noether's theorem this is a conserved quantity of the Hamiltonian system~\eqref{eq:Hamiltonian_system-H}. However, because $\norm{\psi}^{2}$ is conserved, we may define an alternative conserved quantity
\begin{gather*}
 F_{2}\colon \ T^{*}\C^{2} \to \so(2)^{*}\cong \R, \qquad F_{2}(\chi,\psi) \defeq -\frac{\norm{\psi}^{2}}{2}\Im(\psi^{*}\chi).
\end{gather*}
Then we see that $f_{2} \circ \mathbf{M} = F_{2}$ with
\begin{gather*}
 f_{2}\colon \ \se(3)^{*} \to \R, \qquad f_{2}(\bPi,\bGamma) \defeq \bPi \cdot \bGamma.
\end{gather*}
This is the other well-known Casimir of the heavy top bracket~\eqref{eq:PB-heavy_top} or \eqref{eq:PB-se3star}.

\subsection{Symmetry and conserved quantity of the Lagrange top}
Now consider the Lagrange top, i.e., $I_{2} = I_{1}$ and $\mathbf{c} = (0, 0, 1)$. Note that we have
\begin{gather*}
 \parentheses{ \Im\parentheses{ \chi_{1} \bar{\psi}_{2} + \chi_{2} \bar{\psi}_{1} } }^{2}
 + \parentheses{ \Re\parentheses{ \chi_{2} \bar{\psi}_{1} - \chi_{1} \bar{\psi}_{2} } }^{2} \\
 \qquad {}= \abs{ \chi_{1} \bar{\psi}_{2} }^{2} + \abs{ \chi_{2} \bar{\psi}_{1} }^{2} - 2\parentheses{
 \Re\parentheses{ \chi_{1} \bar{\psi}_{2}} \Re\parentheses{ \chi_{2} \bar{\psi}_{1} }
 - \Im\parentheses{ \chi_{1} \bar{\psi}_{2}} \Im\parentheses{ \chi_{2} \bar{\psi}_{1} } } \\
 \qquad{}= \abs{ \chi_{1} \bar{\psi}_{2} }^{2}
 + \abs{ \chi_{2} \bar{\psi}_{1} }^{2}
 - 2\Re\parentheses{ \chi_{1} \bar{\psi}_{2} \chi_{2} \bar{\psi}_{1} },
\end{gather*}
and so the Hamiltonian~\eqref{eq:H} becomes
\begin{gather}
 H_{\rm L}(\chi,\psi) \defeq \frac{1}{8}\parentheses{ \frac{ \abs{ \chi_{1} \bar{\psi}_{2} }^{2} + \abs{ \chi_{2} \bar{\psi}_{1} }^{2} - 2\Re\parentheses{ \chi_{1} \bar{\psi}_{2} \chi_{2} \bar{\psi}_{1} } }{I_{1}} + \frac{ \parentheses{ \Im\parentheses{ \chi_{2} \bar{\psi}_{2} - \chi_{1} \bar{\psi}_{1} } }^{2} }{I_{3}} } \nonumber\\
\hphantom{H_{\rm L}(\chi,\psi)\defeq }{} + m g l \big( |\psi_{1}|^{2} - |\psi_{2}|^{2} \big).\label{eq:H_L}
\end{gather}

Consider the following $\SO(2) \cong \mathbb{S}^{1}$-action on $\C^{2}$:
\begin{gather*}
 \mathbb{S}^{1} \times \C^{2} \to \C^{2}, \qquad \parentheses{ {\rm e}^{\rmi\theta},
 \begin{bmatrix}
 \chi_{1} \\
 \chi_{2}
 \end{bmatrix}
 }
 \mapsto
 \begin{bmatrix}
 {\rm e}^{\rmi\theta} \chi_{1} \\
 {\rm e}^{-\rmi\theta} \chi_{2}
 \end{bmatrix}
 =
 \begin{bmatrix}
 {\rm e}^{\rmi\theta} & 0 \\
 0 & {\rm e}^{-\rmi\theta}
 \end{bmatrix}
 \chi.
\end{gather*}
Its cotangent lift is
\begin{gather*}
 \mathbb{S}^{1} \times T^{*}\C^{2} \to T^{*}\C^{2}, \qquad \big({\rm e}^{\rmi\theta}, (\chi,\psi)\big)
 \mapsto
 \parentheses{ \begin{bmatrix}
 {\rm e}^{\rmi\theta} & 0 \\
 0 & {\rm e}^{-\rmi\theta}
 \end{bmatrix}
 \chi,
 \begin{bmatrix}
 {\rm e}^{\rmi\theta} & 0 \\
 0 & {\rm e}^{-\rmi\theta}
 \end{bmatrix}
 \psi },
\end{gather*}
and we see that the Hamiltonian \eqref{eq:H_L} is invariant under this action.

What is the associated momentum map $J_{3}\colon T^{*}\C^{2} \to \so(2)^{*} \cong \R$? Let any $\omega \in \so(2) \cong \R$ be arbitrary. Its infinitesimal generator is then
\begin{gather*}
 \omega_{\C^{2}}(\chi) = \dzero{s}{
 \begin{bmatrix}
 {\rm e}^{\rmi s\omega} & 0 \\
 0 & {\rm e}^{-\rmi s\omega}
 \end{bmatrix}
 \chi
 }
 = \rmi\omega
 \begin{bmatrix}
 1 & 0 \\
 0 & -1
 \end{bmatrix} \chi.
\end{gather*}
Then we have
\begin{gather*}
 J_{3}(\chi,\psi) \cdot \omega = \Theta_{1}(\chi,\psi) \cdot \omega_{\C^{2}}(\chi)
 = \Re\parentheses{ \psi^{*} \rmi\omega
 \begin{bmatrix}
 1 & 0 \\
 0 & -1
 \end{bmatrix} \chi
 }
 = -\Im\parentheses{
 \psi^{*}
 \begin{bmatrix}
 1 & 0 \\
 0 & -1
 \end{bmatrix} \chi
 } \cdot \omega \\
 \hphantom{J_{3}(\chi,\psi) \cdot \omega}{} = \Im\parentheses{
 \chi_{2} \bar{\psi}_{2} - \chi_{1} \bar{\psi}_{1}
 } \cdot \omega.
\end{gather*}
Therefore, we obtain
\begin{gather*}
 J_{3}(\chi,\psi) = \Im\parentheses{ \chi_{2} \bar{\psi}_{2} - \chi_{1} \bar{\psi}_{1} },
\end{gather*}
and hence
\begin{gather*}
 F_{3}\colon \ T^{*}\C^{2} \to \R, \qquad F_{3}(\chi,\psi) \defeq \frac{1}{2}J_{3}(\chi,\psi)
 = \frac{1}{2} \Im\parentheses{ \chi_{2} \bar{\psi}_{2} - \chi_{1} \bar{\psi}_{1} }
\end{gather*}
is a conserved quantity of \eqref{eq:Hamiltonian_system-H}.
Then we see that $f_{3} \circ \mathbf{M} = F_{3}$ with
\begin{gather*}
 f_{3}\colon \ \se(3)^{*} \to \R, \qquad f_{3}(\bPi,\bGamma) \defeq \Pi_{3}.
\end{gather*}
Again, $f_{3}$ is a well-known conserved quantity of the Lagrange top.

\begin{Remark} It is well known that $h$ (with $I_{2} = I_{1}$ and $\mathbf{c} = (0,0,1)$), $f_{1}$, $f_{2}$, and $f_{3}$ are in involution with respect to the heavy top bracket~\eqref{eq:PB-heavy_top}. Since $\mathbf{M}\colon T^{*}\C^{2} \to \se(3)^{*}$ is Poisson, this implies that $H_{\rm L}$, $F_{1}$, $F_{2}$, and $F_{3}$ are in involution with respect to the canonical Poisson bracket~\eqref{eq:PB-TstarC2} as well.
\end{Remark}

\section{Collective Lie--Poisson integrator for heavy top dynamics}\label{sec:integrator}
\subsection{Collective Lie--Poisson integrator}
The collective formulation suggests that any symplectic integrator for the canonical Hamiltonian system~\eqref{eq:Hamiltonian_system-H-coordinates} on $T^{*}\C^{2} \cong T^{*}\R^{4}$ gives rise to a Lie--Poisson integrator for the heavy top dynamics via the map $\mathbf{M}$. Specifically, let $\Delta t$ be the time step, and $\Phi^{\rm d}_{\Delta t}\colon T^{*}\C^{2} \to T^{*}\C^{2}$ be the discrete flow defined by the symplectic integrator. Then, $\phi^{\rm d}_{\Delta t} \defeq \mathbf{M} \circ \Phi^{\rm d}_{\Delta t}\colon \se(3)^{*} \to \se(3)^{*}$ is clearly Poisson. This is the basic idea of the \textit{collective Lie--Poisson integrator} of \cite{McMoVe2014}; it is implemented for the free rigid body in \cite{McMoVe2015}.

In our case, the symplecticity of the integrator implies that the momentum maps $J_{1}$ and $J_{2}$ from Section~\ref{sec:symmetry_and_conservation_laws} are conserved exactly along the flow $\Phi^{\rm d}_{\Delta t}$. This implies that the corresponding Casimirs $f_{1}$ and $f_{2}$ are exactly conserved along the flow $\phi^{\rm d}_{\Delta t}$ of the collective integrator as well. Furthermore, since the Hamiltonian $H$ is nearly conserved without drifts with symplectic integrators, the Hamiltonian $h$ for the heavy top behaves in a similar manner as well.

We note in passing that \cite{McMoVe2016} showed that the spherical midpoint method~\cite{McMoVe2017} on $\big(\mathbb{S}^{2}\big)^{n}$ is a collective integrator corresponding to the midpoint rule applied to a canonical Hamiltonian system on $T^{*}\R^{2n}$. That is, for this special case, the collective integrator gives rise to an integrator intrinsically defined on the symplectic leaves of a Poisson manifold. It is an interesting future work to look into an extension of this property to our setting.
In our case, each non-trivial coadjoint orbit (symplectic leaf) in $\se(3)^{*}$ is known to be diffeomorphic to either $\mathbb{S}^{2}$ or its tangent bundle~\cite[Section~14.7]{MaRa1999}.

\subsection{Numerical results}
As a test case, we consider the Kovalevskaya top~\cite{Ko1889}~-- the case with $I_{1} = I_{2} = 2I_{3}$ and $\mathbf{c} = (1,0,0)$~-- with the parameters $m = g = l = I_{3} = 1$ and initial condition $\bPi(0) = (2, 3, 4)$ and $\bGamma(0) = \big(1/2, 0, \sqrt{3}/2\big)$.
By solving $\mathbf{M}(\chi(0), \psi(0)) = (\bPi(0), \bGamma(0))$ with the constraint \smash{$\Re(\chi_{1}(0)) = 1$}, we have
\begin{gather*}
 \chi(0) = \parentheses{
 1 - \rmi\big(\sqrt{2} + 3\sqrt{6}\big),
 \frac{ -1 + 12\sqrt{2} + \sqrt{3} - \rmi\,2\sqrt{6} }{1 + \sqrt{3}}
 }, \qquad\!
 \psi(0) = \frac{1}{2\sqrt{2}} \big(\sqrt{3} + 1, \sqrt{3} - 1\big)
\end{gather*}
as a corresponding initial condition for \eqref{eq:Hamiltonian_system-H} or \eqref{eq:Hamiltonian_system-H-coordinates}.

The system has four conserved quantities: the Hamiltonian $h$, the Casimirs $f_{1}$ and $f_{2}$, and the \textit{Kovalevskaya invariant}
\begin{gather*}
 K(\bPi,\bGamma) \defeq \big| (\Pi_{1} + \rmi\Pi_{2})^{2} - 4 m g l I_{3}(\Gamma_{1} + \rmi\Gamma_{2})\big|^{2},
\end{gather*}
and is known to be integrable~\cite{Ko1889}; see also \cite{Au1998}. Our focus here is to compare the behaviors of these four conserved quantities for several heavy top integrators.

Fig.~\ref{fig:comparison} shows the time evolutions of the four conserved quantities along the numerical solutions obtained by applying the explicit and implicit midpoint rules to the canonical Hamiltonian system~\eqref{eq:Hamiltonian_system-H} as well as that obtained by applying the explicit midpoint rule directly to the heavy top equations~\eqref{eq:heavy_top}. The time step $\Delta t$ is $1/50$ in all the cases. We note that the implicit midpoint rule is a symplectic integrator for canonical Hamiltonian systems (see, e.g., \cite[Theorem~VI.3.5]{HaLuWa2006} and \cite[Section~4.1]{LeRe2004}). So the implicit midpoint rule applied to \eqref{eq:Hamiltonian_system-H} gives a collective Lie--Poisson integrator for the heavy top equations~\eqref{eq:heavy_top} via the map $\mathbf{M}$, and is our main focus here.

\begin{figure}[t] \centering\small
 \includegraphics[width=115mm]{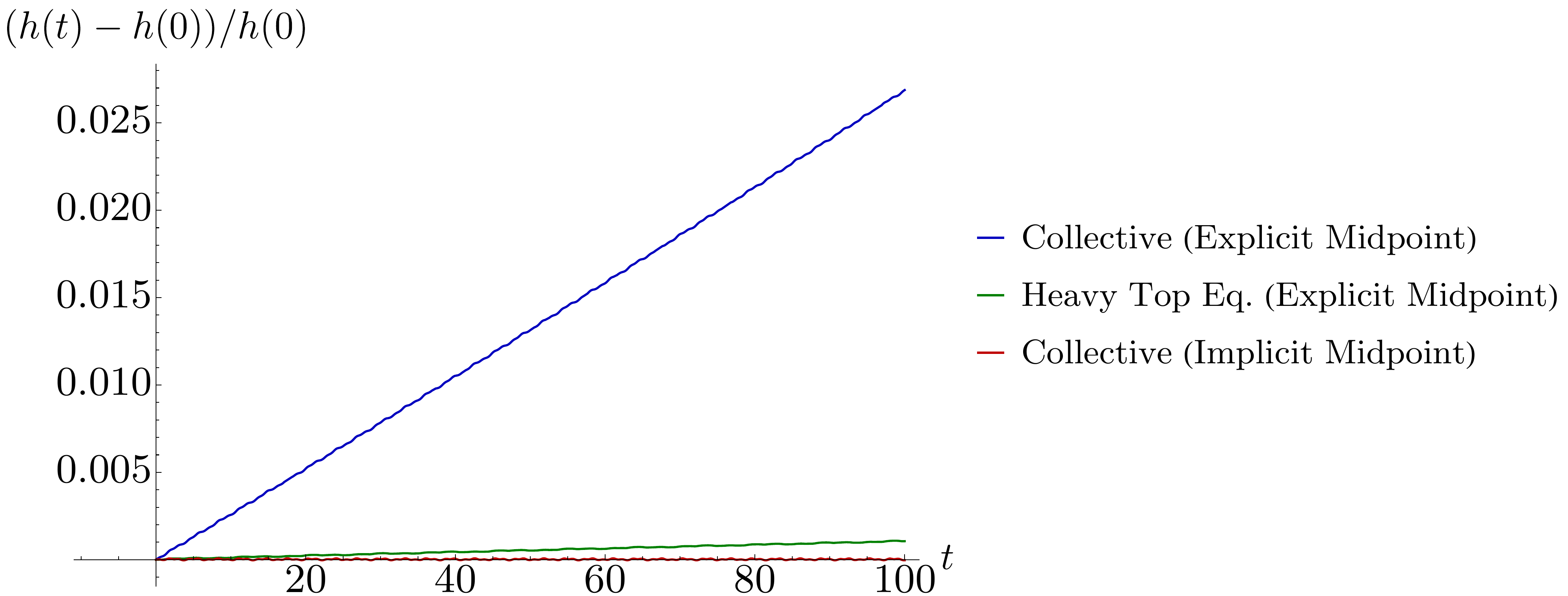}\\
 (a) Hamiltonian $h$

 \vspace{4mm}

\begin{minipage}[b]{70mm}\centering\small
\includegraphics[width=65mm]{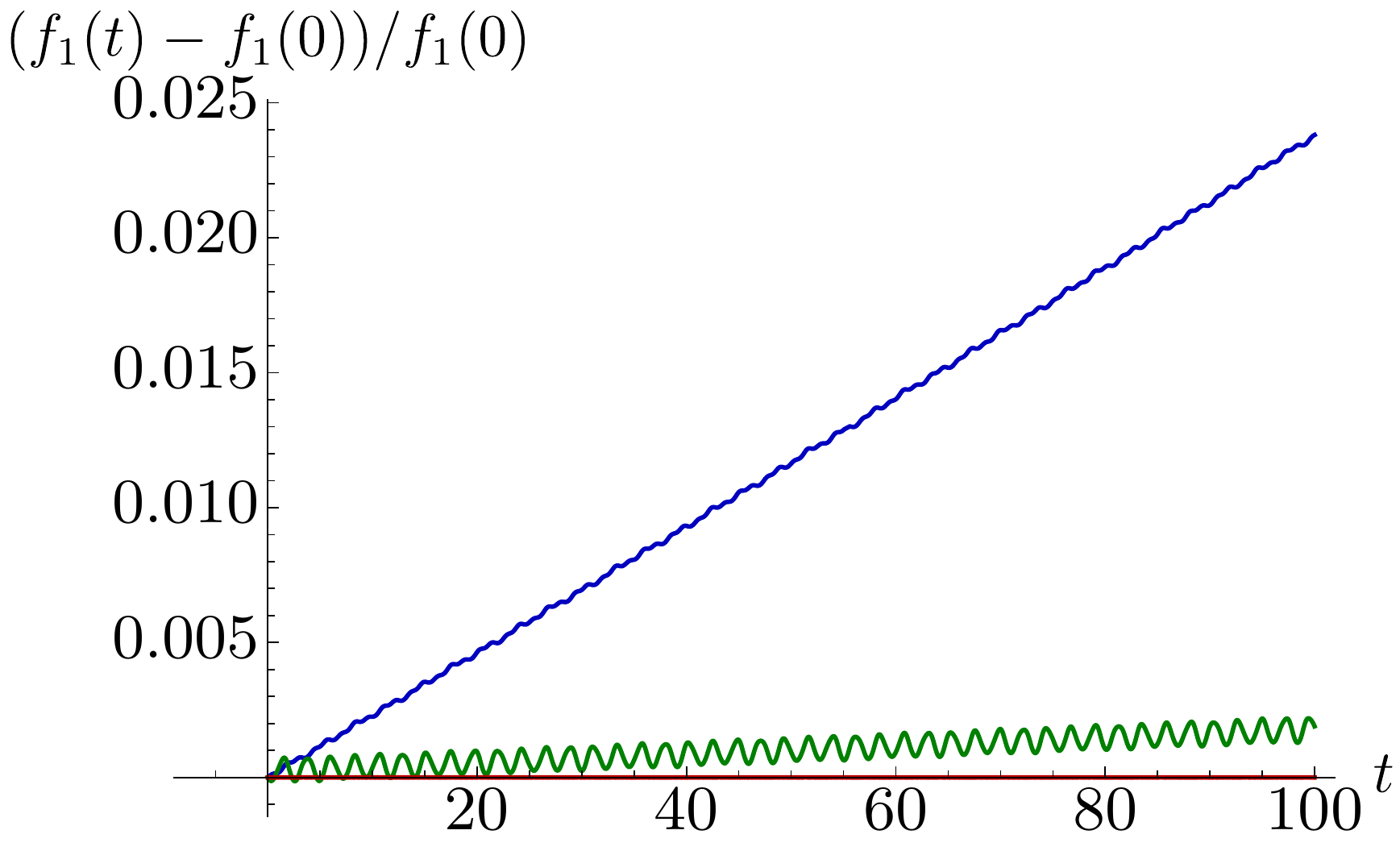}\\
(b) Casimir $f_{1}$
\end{minipage}
\quad
\begin{minipage}[b]{70mm}\centering\small
\includegraphics[width=65mm]{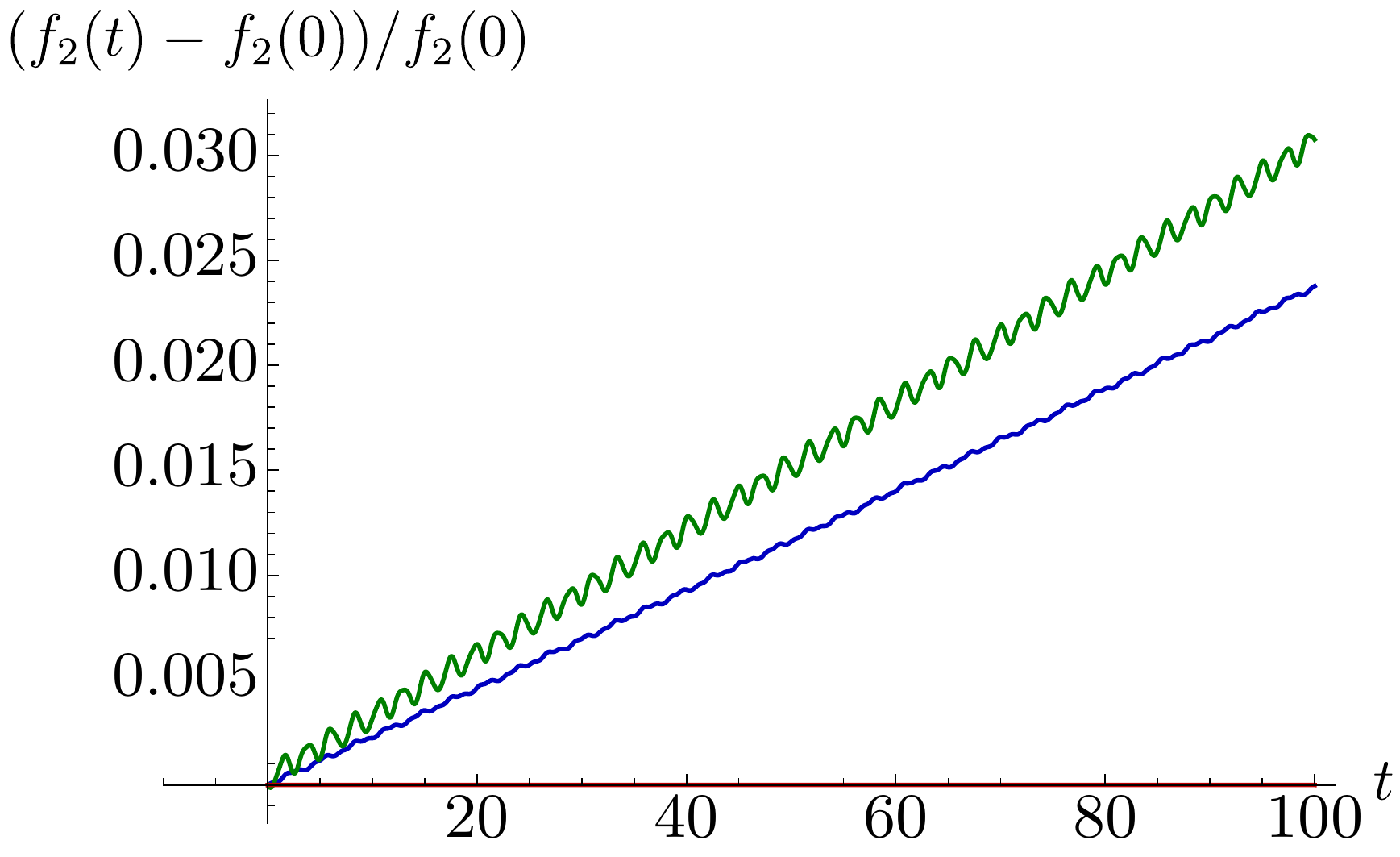}\\
(c) Casimir $f_{2}$
\end{minipage}

\vspace{4mm}

\includegraphics[width=65mm]{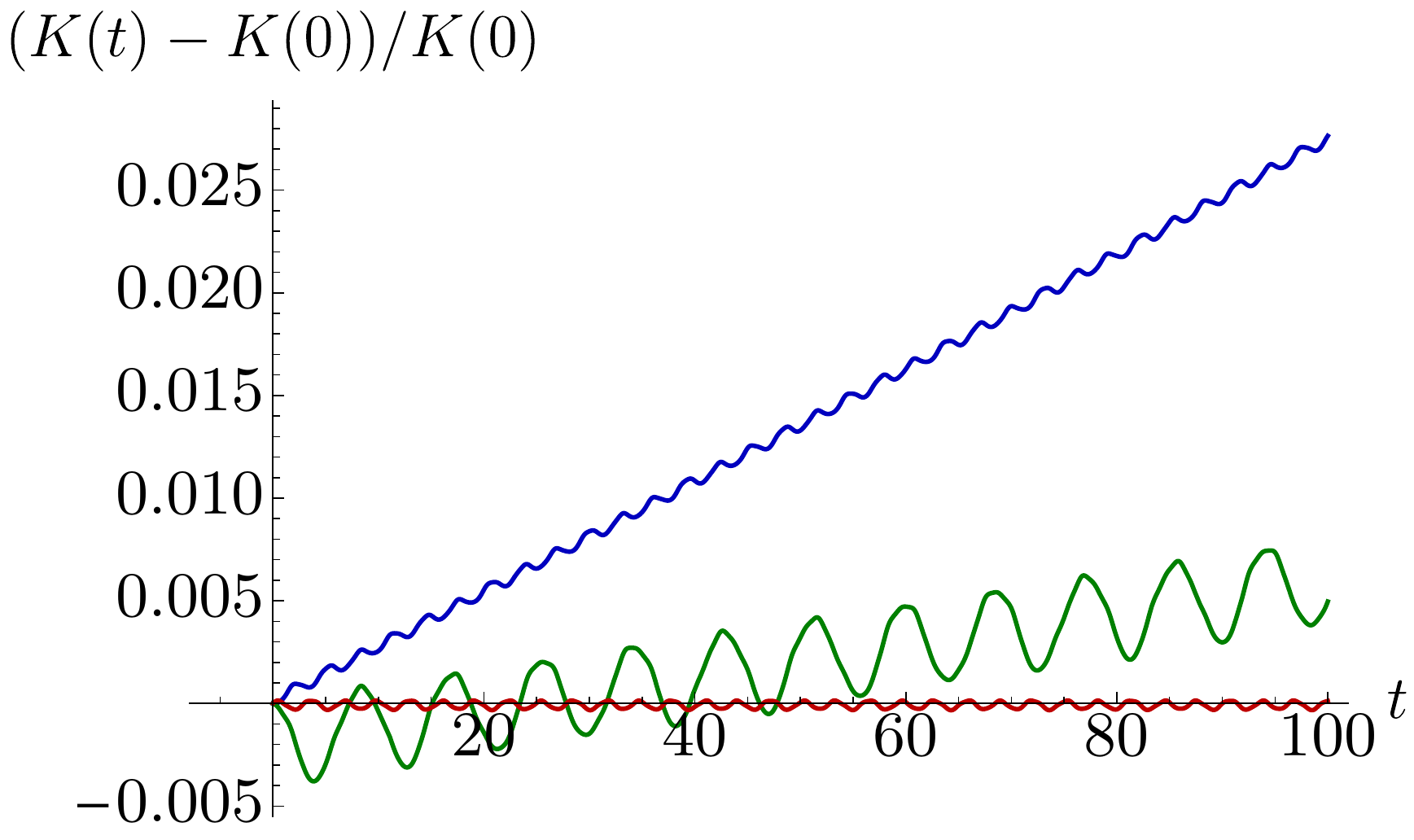}\\
 (d) Kovalevskaya invariant $K$

 \caption{Behaviors of the Hamiltonian and the Casimirs of the Kovalevskaya top with $m = g = l = 1$, $I_{1} = 2$, $I_{2} = 2$, $I_{3} = 1$, $\mathbf{c} = (1, 0, 0)$ for three integrators. Blue: explicit midpoint rule (non-symplectic) applied to~\eqref{eq:Hamiltonian_system-H-coordinates}. Green: explicit midpoint rule (non-Poisson) applied to~\eqref{eq:heavy_top}. Red: implicit midpoint rule (symplectic) applied to~\eqref{eq:Hamiltonian_system-H-coordinates}. $\Delta t = 1/50$, $\bPi(0) = (2, 3, 4)$, $\bGamma(0) = \big(1/2, 0, \sqrt{3}/2\big)$.} \label{fig:comparison}
\end{figure}

The explicit midpoint rule applied to either \eqref{eq:heavy_top} or \eqref{eq:Hamiltonian_system-H} suffers from fairly significant drifts and/or oscillations in all the four conserved quantities. On the other hand, the collective Lie--Poisson integrator conserves the Casimirs $f_{1}$ and $f_{2}$ exactly because it is a symplectic integrator for canonical Hamiltonian systems; note that we identified these Casimirs as Noether conserved quantities, i.e., the momentum maps associated with the symmetries of the canonical Hamiltonian system~\eqref{eq:Hamiltonian_system-H}. One can observe this in Fig.~\ref{fig:comparison}(b) and~(c) as well.
Furthermore, as is well known (see, e.g., \cite[Section~IX.8]{HaLuWa2006} and \cite[Section~5.2.2]{LeRe2004}), symplectic integrators exhibit near-conservation of the Hamiltonian, and one observes this in Fig.~\ref{fig:comparison}(a).

Interestingly, the collective Lie--Poisson integrator exhibits a similar near-conservation of the Kovalevskaya invariant as well, performing better than the non-symplectic ones; see Fig.~\ref{fig:comparison}(d). It is not clear if the Kovalevskaya invariant $K$ corresponds to any symmetry of the canonical Hamiltonian system~\eqref{eq:Hamiltonian_system-H}, but the oscillation in $K$ with amplitudes way beyond the machine precision suggests that $K$ is \textit{not} a Noether conserved quantity of \eqref{eq:Hamiltonian_system-H}.

We also applied the implicit midpoint rule directly to the heavy top equations~\eqref{eq:heavy_top} and compared the results with those of the collective Lie--Poisson integrator as well. The former integrator is known to conserve the Hamiltonian $h$ and the Casimirs $f_{1}$ and $f_{2}$ exactly but is \textit{not} Poisson; see \cite{AuKrWa1993}.
Fig.~\ref{fig:comparison_IMP} shows the behaviors of $h$ and $K$ for the implicit midpoint rule applied to \eqref{eq:heavy_top} or \eqref{eq:Hamiltonian_system-H}.
\begin{figure}[t] \centering
\begin{minipage}[b]{70mm}\centering\small
\includegraphics[width=65mm]{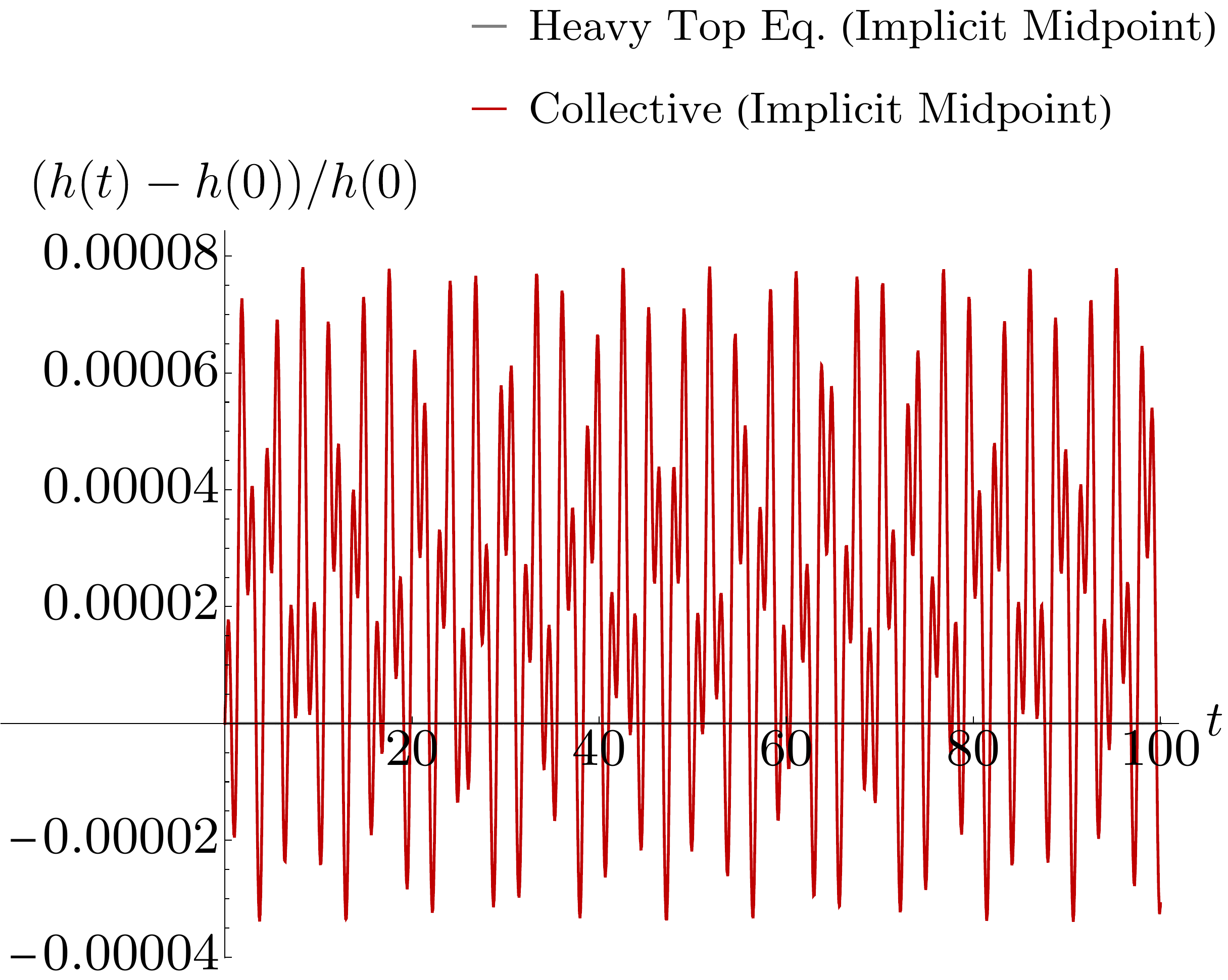}\\
 (a) Hamiltonian $h$
 \end{minipage}
 \quad
 \begin{minipage}[b]{70mm}\centering\small
 \includegraphics[width=65mm]{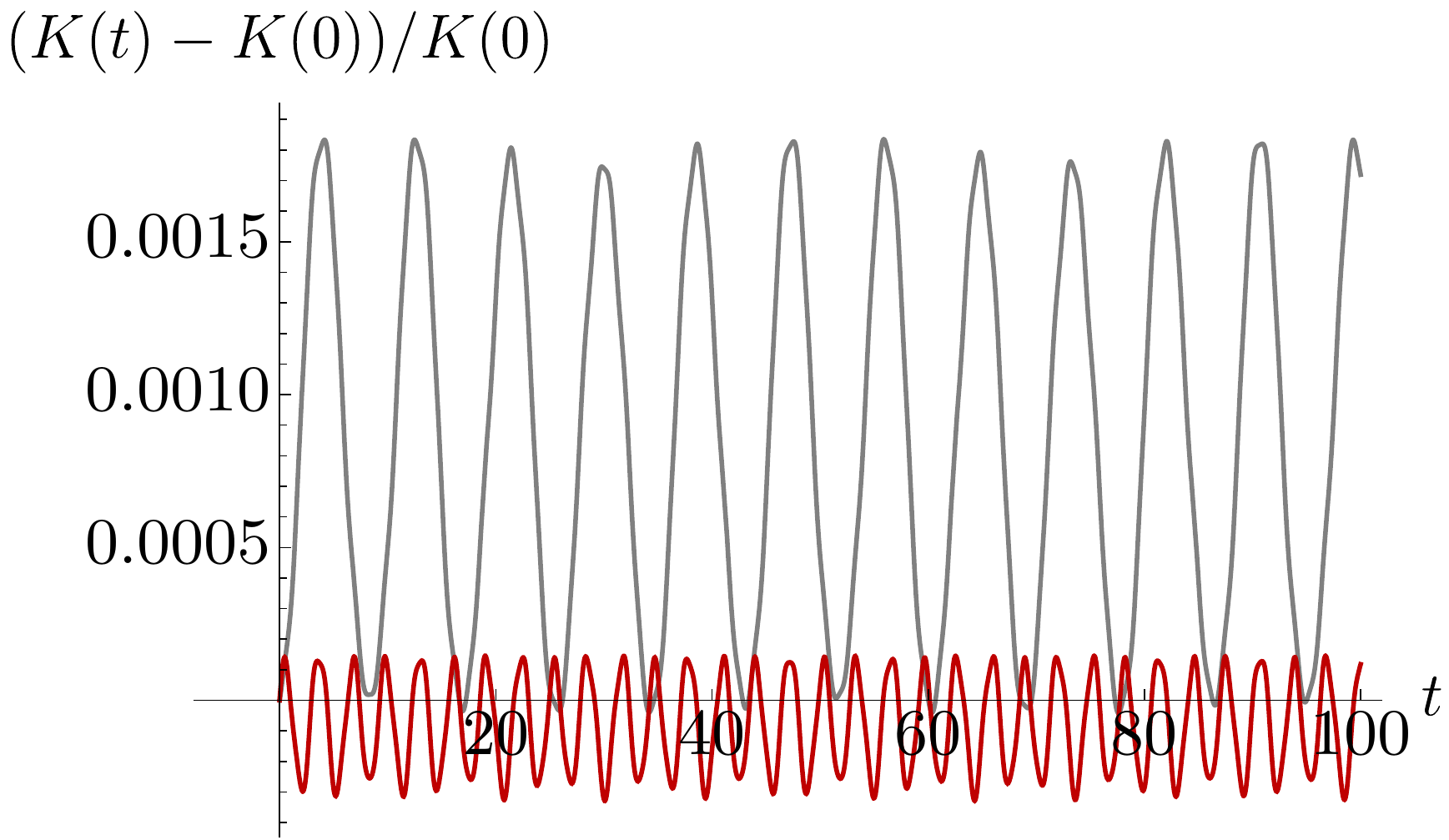}\\
 (b) Kovalevskaya invariant $K$
 \end{minipage}
\caption{Behaviors of the Hamiltonian $h$ and the Kovalevskaya invariant $K$. Gray: implicit midpoint rule (non-Poisson) applied to \eqref{eq:heavy_top}; see \cite{AuKrWa1993}.
 Red: implicit midpoint rule (symplectic) applied to \eqref{eq:Hamiltonian_system-H-coordinates}. The parameters and the time step are the same as in Fig.~\ref{fig:comparison}.}
 \label{fig:comparison_IMP}
\end{figure}

\looseness=1 As mentioned above, the former conserves $h$ exactly, whereas the collective Lie--Poisson integrator exhibits a near-conservation of it with oscillations of very small amplitudes. On the other hand, one observes near-conservation of $K$ in both integrators, but the non-Poisson integrator exhibits an oscillation with a larger amplitude than the collective Lie--Poisson integrator does.

\appendix
\section[Identification of $\su(2)$ with $\so(3)$ and $\R^{3}$]{Identification of $\boldsymbol{\su(2)}$ with $\boldsymbol{\so(3)}$ and $\boldsymbol{\R^{3}}$}\label{sec:su2}
This section gives a brief description of the well-known identification of $\su(2)$ with $\so(3)$ and~$\R^{3}$. The purpose is to introduce the notation used throughout the paper.

\subsection[Isomorphisms between $\su(2)$, $\so(3)$, and $\R^{3}$]{Isomorphisms between $\boldsymbol{\su(2)}$, $\boldsymbol{\so(3)}$, and $\boldsymbol{\R^{3}}$}
Let $\{ \mathbf{e}_{i} \}_{i=1}^{3}$ be the standard basis for $\R^{3}$, and define a basis $\{ e_{i} \}_{i=1}^{3}$ for $\su(2)$ as
\begin{gather*}
 e_{1} \defeq -\frac{\rmi}{2}
 \begin{bmatrix}
 0 & 1 \\
 1 & 0
 \end{bmatrix},
 \qquad
 e_{2} \defeq
 -\frac{\rmi}{2}
 \begin{bmatrix}
 0 & -\rmi \\
 \rmi & 0
 \end{bmatrix},
 \qquad
 e_{3} \defeq
 -\frac{\rmi}{2}
 \begin{bmatrix}
 1 & 0 \\
 0 & -1
 \end{bmatrix},
\end{gather*}
as well as a basis $\{ \hat{e}_{i} \}_{i=1}^{3}$ for $\so(3)$ as
\begin{gather*}
 \hat{e}_{1} \defeq
 \begin{bmatrix}
 0 & 0 & 0 \\
 0 & 0 & -1 \\
 0 & 1 & 0
 \end{bmatrix},
 \qquad
 \hat{e}_{2} \defeq
 \begin{bmatrix}
 0 & 0 & 1 \\
 0 & 0 & 0 \\
 -1 & 0 & 0
 \end{bmatrix},
 \qquad
 \hat{e}_{3} \defeq
 \begin{bmatrix}
 0 & -1 & 0 \\
 1 & 0 & 0 \\
 0 & 0 & 0
 \end{bmatrix}.
\end{gather*}
So we have the isomorphisms
\begin{gather*}
 \su(2) \longleftrightarrow \so(3) \longleftrightarrow \R^{3}
\end{gather*}
defined by
\begin{gather*}
 e_{i} \longleftrightarrow \hat{e}_{i} \longleftrightarrow \mathbf{e}_{i}.
\end{gather*}
It is well known that these are Lie algebra isomorphisms in the sense that $[e_{i},e_{j}]$, $[\hat{e}_{i},\hat{e}_{j}]$, and $\mathbf{e}_{i} \times \mathbf{e}_{j}$ with $i,j \in \{1, 2, 3\}$ are all the same under the identification.

As a result, we may identify $\su(2)$ with $\R^{3}$ via the map
\begin{gather*}
 \R^{3} \to \su(2), \qquad
 \boldsymbol{\xi} = (\xi_{1}, \xi_{2}, \xi_{3})
 \mapsto
 \xi \defeq \sum_{j=1}^{3} \xi_{j} e_{j} = -\frac{\rmi}{2}
 \begin{bmatrix}
 \xi_{3} & \xi_{1} - \rmi\xi_{2} \\
 \xi_{1} + \rmi\xi_{2} & -\xi_{3}
 \end{bmatrix},
\end{gather*}
and also $\so(3)$ with $\R^{3}$ via the following ``hat map''
\begin{gather*}
 \hat{(\,\cdot\,)} \colon \ \R^{3} \to \so(3), \qquad
 \boldsymbol{\xi} = (\xi_{1}, \xi_{2}, \xi_{3}) \mapsto \hat{\xi} \defeq \sum_{j=1}^{3} \xi_{j} \hat{e}_{j}
 = \begin{bmatrix}
 0 & -\xi_{3} & \xi_{2} \\
 \xi_{3} & 0 & -\xi_{1} \\
 -\xi_{2} & \xi_{1} & 0
 \end{bmatrix}.
\end{gather*}

\subsection{Inner products}\label{ssec:ips}
Let us define inner products on $\su(2)$ and $\so(3)$ as follows. For any $\boldsymbol{\xi}, \boldsymbol{\eta} \in \R^{3}$, using the notation from above,
\begin{gather} \label{eq:ip-su2}
 \tip{ \xi }{ \eta }_{\su(2)} \defeq 2\tr(\xi^{*} \eta)
\end{gather}
as well as
\begin{gather} \label{eq:ip-so3}
 \big\langle \hat{\xi} , \hat{\eta}\big\rangle_{\so(3)} \defeq \frac{1}{2}\tr\big( \hat{\xi}^{T}\hat{\eta}\big).
\end{gather}
It is easy to check that these are compatible with the standard dot product in $\R^{3}$ under the above identification, i.e.,
\begin{gather*}
 \tip{ \xi }{ \eta }_{\su(2)} = \big\langle \hat{\xi}, \hat{\eta}\big\rangle_{\so(3)} = \boldsymbol{\xi} \cdot \boldsymbol{\eta}.
\end{gather*}
Using these inner products we may identify the duals $\su(2)^{*}$ and $\so(3)^{*}$ with $\su(2)$ and $\so(3)$ respectively, and in turn, with $\R^{3}$ as well.

\section{Proof of Proposition~\ref{prop:varpi}}\label{sec:prop:varpi-proof}
Let us first decompose the Poisson bracket \eqref{eq:PB-gstar} into two to define, for any smooth $\tilde{f}, \tilde{h}\colon \big(\su(2) \ltimes \C^{2}\big)^{*} \to \R$,
\begin{gather*}
 \big\{\tilde{f},\tilde{h}\big\}_{\su(2)^{*}} (\mu, \alpha) \defeq -\ip{\mu}{ \left[ \fd{\tilde{f}}{\mu}, \fd{\tilde{h}}{\mu} \right] }_{\su(2)},\\
 \big\{\tilde{f},\tilde{h}\big\}_{(\C^{2})^{*}} (\mu, \alpha) \defeq -\ip{\alpha}{ \fd{\tilde{f}}{\mu} \fd{\tilde{h}}{\alpha} - \fd{\tilde{h}}{\mu} \fd{\tilde{f}}{\alpha} }_{\C^{2}}.
\end{gather*}
Likewise, we decompose \eqref{eq:PB-se3star} to define, for any smooth $f, h\colon \se(3)^{*} \to \R$,
\begin{gather*}
 \PB{f}{h}_{\so(3)^{*}} \defeq -\ip{\hat{\Pi}}{ \left[ \fd{f}{\hat{\Pi}}, \fd{h}{\hat{\Pi}} \right] }_{\so(3)},
 \qquad
 \PB{f}{h}_{(\R^{3})^{*}} \defeq - \bGamma \cdot \parentheses{ \fd{f}{\hat{\Pi}} \pd{h}{\bGamma} - \fd{h}{\hat{\Pi}} \pd{f}{\bGamma} }.
\end{gather*}
Our goal is to show that
\begin{gather}
 \label{eq:su2-so3_Poisson}
 \PB{f \circ \varpi}{h \circ \varpi}_{\su(2)^{*}}
 = \PB{f}{h}_{\so(3)^{*}} \circ \varpi
\end{gather}
as well as
\begin{gather}
 \label{eq:C2-R3_Poisson}
 \PB{f \circ \varpi}{h \circ \varpi}_{(\C^{2})^{*}} = \PB{f}{h}_{(\R^{3})^{*}} \circ \varpi.
\end{gather}

Let us first show \eqref{eq:su2-so3_Poisson}.
Let $(\mu,\alpha) \in (\su(2) \ltimes \C^{2})^{*}$ and $\delta\mu \in \su(2)^{*}$ be arbitrary.
Then, setting $\mathbf{a} \defeq \varpi_{2}(\alpha)$, we have
\begin{align*}
 \ip{ \delta\mu }{ \fd{(f \circ \varpi)}{\mu} }_{\su(2)}
 &= \dzero{s}{ f \circ \varpi(\mu + s\delta\mu, \alpha) } \\
 &= \dzero{s}{ f\parentheses{ \hat{\mu} + s\widehat{\delta\mu}, \mathbf{a} } }= \ip{ \widehat{\delta\mu} }{ \fd{f}{\hat{\mu}} }_{\so(3)}.
\end{align*}
This implies that we may identify $\tfd{(f \circ \varpi)}{\mu} \in \su(2)^{*}$ and $\tfd{f}{\hat{\mu}} \in \so(3)^{*}$; the same goes with the derivatives of $h$.
The compatibility of the inner products in $\su(2)$ and $\so(3)$ (see Ap\-pen\-dix~\ref{ssec:ips}) then yields
\begin{align*}
 \PB{f \circ \varpi}{h \circ \varpi}_{\su(2)^{*}} (\mu, \alpha)
 &= -\ip{\mu}{ \left[ \fd{(f \circ \varpi)}{\mu}, \fd{(h \circ \varpi)}{\mu} \right] }_{\su(2)} \\
 &= -\ip{\hat{\mu}}{ \left[ \fd{f}{\hat{\mu}}, \fd{h}{\hat{\mu}} \right] }_{\so(3)}= \PB{f}{h}_{\so(3)^{*}} \circ \varpi(\mu, \alpha).
\end{align*}

So it remains to show \eqref{eq:C2-R3_Poisson}. Writing $\xi \defeq \tfd{(f \circ \varpi)}{\mu}$, $\eta \defeq \tfd{(h \circ \varpi)}{\mu}$, $\beta \defeq \tfd{(f \circ \varpi)}{\alpha}$, and $\gamma \defeq \tfd{(h \circ \varpi)}{\alpha}$ for short,
\begin{align*}
 \PB{f \circ \varpi}{h \circ \varpi}_{(\C^{2})^{*}} (\mu, \alpha)
 &= -\Re\parentheses{ \alpha^{*} ( \xi \gamma - \eta \beta ) }
 = -\frac{1}{2}\parentheses{ \alpha^{*} \xi \gamma + \gamma^{*} \xi^{*} \alpha - \alpha^{*} \eta \beta - \beta^{*} \eta^{*} \alpha } \\
 &= -\frac{1}{2}\tr\parentheses{ (\gamma \alpha^{*} - \alpha \gamma^{*}) \xi } + \frac{1}{2}\tr\parentheses{ (\beta \alpha^{*} - \alpha \beta^{*}) \eta } \\
 &= -\ip{ \frac{1}{4}(\alpha \gamma^{*} - \gamma \alpha^{*}) }{ \xi }_{\su(2)}
 + \ip{ \frac{1}{4}(\alpha \beta^{*} - \beta \alpha^{*}) }{ \eta }_{\su(2)}.
\end{align*}
Let us find an expression for $\beta = \tfd{(f \circ \varpi)}{\alpha} \in \C^{2}$. For any $\delta\alpha \in \big(\C^{2}\big)^{*}$,
\begin{align*}
 \ip{ \delta\alpha }{ \fd{(f \circ \varpi)}{\alpha} }_{\C^{2}}
 &= \dzero{s}{ f \circ \varpi(\mu, \alpha + s\delta\alpha) }
 = \dzero{s}{ f\parentheses{ \hat{\mu}, \varpi_{2}(\alpha + s\delta\alpha) } } \\
 & = \nabla_{\mathbf{a}}f(\hat{\mu}, \mathbf{a}) \cdot T_{\alpha}\varpi_{2}(\delta\alpha),
\end{align*}
where $\nabla_{\mathbf{a}}$ stands for the gradient with respect to $\mathbf{a}$. However, using the expression~\eqref{eq:varpi2} for~$\varpi_{2}$, we obtain the tangent map $T_{\alpha}\varpi_{2}$ as follows
\begin{gather*}
 T_{\alpha}\varpi_{2}(\delta\alpha) =
 2 \begin{bmatrix}
 \Re\parentheses{ \overline{\delta\alpha}_{1} \alpha_{2} } + \Re\parentheses{ \overline{\alpha}_{1} \delta\alpha_{2} } \vspace{1mm}\\
 \Im\parentheses{ \overline{\delta\alpha}_{1} \alpha_{2} } - \Im\parentheses{ \overline{\alpha}_{1} \delta\alpha_{2} } \vspace{1mm}\\
 \Re\parentheses{ \alpha_{1} \overline{\delta\alpha}_{1} } - \Re\parentheses{ \alpha_{2} \overline{\delta\alpha_{2}} }
 \end{bmatrix}
 = 2 \begin{bmatrix}
 \Re\parentheses{ \alpha_{2} \overline{\delta\alpha}_{1} } + \Re\parentheses{ \alpha_{1} \overline{\delta\alpha}_{2} } \vspace{1mm}\\
 -\Re\parentheses{ \rmi\alpha_{2} \overline{\delta\alpha}_{1} } + \Re\parentheses{ \rmi\alpha_{1} \overline{\delta\alpha}_{2} } \vspace{1mm}\\
 \Re\parentheses{ \alpha_{1} \overline{\delta\alpha}_{1} } - \Re\parentheses{ \alpha_{2} \overline{\delta\alpha}_{2} }
 \end{bmatrix}.
\end{gather*}
Therefore,
\begin{gather*}
 \ip{ \delta\alpha }{ \fd{(f \circ \varpi)}{\alpha} }_{\C^{2}}
 = 2\Re\parentheses{ \overline{\delta\alpha}_{1} \parentheses{ \alpha_{2} \pd{f}{a_{1}} - \rmi\alpha_{2} \pd{f}{a_{2}} + \alpha_{1} \pd{f}{a_{3}} } } \\
\hphantom{\ip{ \delta\alpha }{ \fd{(f \circ \varpi)}{\alpha} }_{\C^{2}}=}{}
 + 2\Re\parentheses{ \overline{\delta\alpha}_{2} \parentheses{ \alpha_{1} \pd{f}{a_{1}} + \rmi\alpha_{1} \pd{f}{a_{2}} - \alpha_{2} \pd{f}{a_{3}} } },
\end{gather*}
and so we have
\begin{gather*}
 \fd{(f \circ \varpi)}{\alpha}
 = 2 \begin{bmatrix}
 \DS \alpha_{2} \pd{f}{a_{1}} - \rmi\alpha_{2} \pd{f}{a_{2}} + \alpha_{1} \pd{f}{a_{3}} \vspace{1mm}\\
 \DS \alpha_{1} \pd{f}{a_{1}} + \rmi\alpha_{1} \pd{f}{a_{2}} - \alpha_{2} \pd{f}{a_{3}}
 \end{bmatrix}.
\end{gather*}
Then a tedious but straightforward calculation yields
\begin{gather*}
 \beta = \frac{1}{4}(\alpha\beta^{*} - \beta \alpha^{*})
 = \parentheses{ a_{3}\pd{f}{a_{2}} - a_{2}\pd{f}{a_{3}},\, a_{1}\pd{f}{a_{3}} - a_{3}\pd{f}{a_{1}},\, a_{2}\pd{f}{a_{1}} - a_{1}\pd{f}{a_{2}} }
 = \nabla_{\mathbf{a}}f \times \mathbf{a}
\end{gather*}
under the identification $\su(2)^{*} \cong \R^{3}$. Similarly,
\begin{gather*}
 \frac{1}{4}(\alpha\gamma^{*} - \gamma \alpha^{*}) = \nabla_{\mathbf{a}}h \times \mathbf{a}.
\end{gather*}
Therefore, we obtain
\begin{align*}
 \PB{f \circ \varpi}{h \circ \varpi}_{(\C^{2})^{*}} (\mu, \alpha)
 &= -(\nabla_{\mathbf{a}}h \times \mathbf{a}) \cdot \boldsymbol{\xi}
 + (\nabla_{\mathbf{a}}f \times \mathbf{a}) \cdot \boldsymbol{\eta} = -\mathbf{a} \cdot (\boldsymbol{\xi} \times \nabla_{\mathbf{a}}h - \boldsymbol{\eta} \times \nabla_{\mathbf{a}}f) \\
 &= -\mathbf{a} \cdot \big(\hat{\xi} \nabla_{\mathbf{a}}h - \hat{\eta} \nabla_{\mathbf{a}}f\big)
 = -\mathbf{a} \cdot \parentheses{ \fd{f}{\hat{\mu}} \fd{h}{\mathbf{a}} - \fd{h}{\hat{\mu}} \fd{f}{\mathbf{a}} } \\
 & = \PB{f}{h}_{(\R^{3})^{*}} \circ \varpi(\mu, \alpha).
\end{align*}

\subsection*{Acknowledgments}
I would like to thank the referees for their helpful comments and suggestions. This work was partially supported by NSF grant CMMI-1824798.

\pdfbookmark[1]{References}{ref}
\LastPageEnding

\end{document}